\documentclass[11pt]{article}

\usepackage{fullpage}
\usepackage[utf8]{inputenc}
\usepackage{todonotes}
\usepackage[linesnumbered,ruled,vlined]{algorithm2e}
\SetKw{Break}{break}
\SetKw{AND}{and}
\SetKw{OR}{or}

\usepackage{graphics}
\usepackage[dvips]{epsfig}
\usepackage{comment}
\usepackage{xcolor}
\usepackage{amsmath}
\usepackage{amssymb}
\usepackage{amsfonts}
\usepackage{graphicx, xspace}

\newtheorem{theorem}{Theorem}[section]
\newtheorem{lemma}{Lemma}[section]
\newtheorem{corollary}{Corollary}[section]
\newtheorem{claim}{Claim}[section]

\newcommand{\qed}{\hfill $\Box$ \bigbreak}
\newenvironment{proof}{\noindent {\bf Proof.}}{\qed}

\newenvironment{proofclaim}{\noindent{\bf Proof of the claim.}}{\hfill$\star$}

\title{{\bf Graph Exploration: The Impact of a Distance Constraint}}

\author{
St\'ephane Devismes\thanks{
MIS Lab., Universit\'{e} de Picardie Jules Verne, France. E-mail: stephane.devismes@u-picardie.fr}
\and
Yoann Dieudonn\'{e}\thanks{
MIS Lab., Universit\'{e} de Picardie Jules Verne, France. E-mail: yoann.dieudonne@u-picardie.fr}
\and
Arnaud Labourel\thanks{
Aix Marseille Univ, CNRS, LIS, Marseille, France. Email: arnaud.labourel@lis-lab.fr}
}
\date{ }

\begin{document}
\sloppy
\maketitle
\begin{abstract}
A mobile agent, starting from a node $s$ of a simple undirected
connected graph $G=(V,E)$, has to explore all nodes and edges of $G$
using the minimum number of edge traversals. To do so, the agent uses a deterministic algorithm that allows it to gain information on $G$ as it traverses its edges. 
During its exploration, the agent must
always respect the constraint of knowing a path of length at most $D$
to go back to node $s$. The upper bound $D$ is fixed as being equal to
$(1+\alpha)r$, where $r$ is the eccentricity of node $s$ (i.e., the
maximum distance from $s$ to any other node) and $\alpha$ is any
positive real constant. This task has been introduced by Duncan et
al.~\cite{DuncanKK06} and is known as \emph{distance-constrained
exploration}.

The \emph{penalty} of an exploration algorithm running in $G$ is the
number of edge traversals made by the agent in excess of $|E|$. In
\cite{PanaiteP99}, Panaite and Pelc gave an algorithm for solving
exploration without any constraint on the moves that is guaranteed to
work in every graph $G$ with a (small) penalty in
$\mathcal{O}(|V|)$. Hence, a natural question is whether we could
obtain a distance-constrained exploration algorithm with the same
guarantee as well.

In this paper, we provide a negative answer to this question. We also
observe that an algorithm working in every graph $G$ with a linear
penalty in $|V|$ cannot be obtained for the task of
\emph{fuel-constrained exploration}, another variant studied in the
literature.

This solves an open problem posed by Duncan et
al. in~\cite{DuncanKK06} and shows a fundamental separation with the
task of exploration without constraint on the moves.
\vspace*{1cm}

\noindent{\bf Keywords:} exploration, graph, mobile agent.

\vspace*{2cm}
\end{abstract}

\pagebreak

\section{Introduction}

\subsection{Background}

Exploring an unknown environment is a fundamental task with numerous applications, including target localization, data gathering, and map construction. The considered environments can vary considerably in type: they may be terrains, networks, three-dimensional spaces, etc. In hazardous situations, it is generally more appropriate to delegate exploration to a mobile artificial entity, hereinafter referred to as an agent, which could be an autonomous robot or a vehicle remotely controlled by a human.

In this paper, we focus our attention on the task of exploration by a mobile agent of a network modeled as a graph, and more specifically on a variant where the agent must always know a path from its initial node to its current node that does not exceed a given upper bound. This requirement finds its justification in applications where the agent needs to always keep the possibility of returning quickly to its base or to always maintain a certain distance to ensure ongoing communication with its base.

\subsection{Model and Problem Definition}\label{sec:model}

An agent is assigned the task of exploring a simple\footnote{In \cite{DuncanKK06}, where the task of distance-constrained exploration that we will study has been introduced, the authors do not mention explicitely whether they consider simple graphs or multigraphs. A closer look at their work can reveal, though, that their results hold for multigraphs. However, it is important to note that, in the context of the impossibility proof that we will conduct, the choice of focusing only on simple graphs does not lead to a loss of generality. In fact, it broadens the scope of our result, ensuring that it holds for both simple graphs and multigraphs.} finite undirected graph $G=(V,E)$, starting from a node $s$, called the \emph{source node}. The exploration requires visiting (resp. traversing) at least once every node (resp. every
edge). Thus, $G$ is supposed to be connected. We denote by ${\tt deg}(v)$ the degree of a node $v$ in $G$, and by $|V|$ (resp. $|E|$) the order (resp. size) of $G$.

We make the same assumption as in \cite{DuncanKK06} which
states that the agent has unbounded memory and can recognize already
visited nodes and traversed edges.
% during its exploration.
This is formalized as follows. Nodes of $G$ have arbitrary pairwise
distinct labels that are positive integers. The label of $s$ will be denoted by $l_s$. Distinct labels, called
{\em port numbers} and ranging from $0$ to ${\tt deg}(v)-1$, are
arbitrarily assigned locally at each node $v$ to each of its incident
edge. Hence, each edge has two ports, one per endpoint, and there
is no a priori relationship between these two ports. At the beginning, the
agent located at node $s$
% sees and
learns only its degree and its label. To explore the graph, the
agent applies a deterministic algorithm that makes it act in steps: at each step,
the algorithm selects a port number $p$ at the current node $v$ (on
the basis of all information that has been memorized so far by the
agent) and then asks the agent to traverse the edge having port $p$ at
$v$. At that point, the agent traverses the given edge at any finite positive speed and, once it enters the adjacent node, it learns (only) its
degree, its label as well as the incoming port number. Afterwards, the agent starts
the next step, unless the algorithm has terminated.

Roughly speaking, the memory of the agent is the total information it has collected since the beginning
of its exploration. We formalize it as follows: the
memory of the agent, after the first $t$ edge traversals, is the
sequence $(M_0,M_1,\dots ,M_t)$, where $M_0$ is the information of the
agent at its start and $M_i$ is the information acquired right after
its $i$th edge traversal. Precisely, the initial memory $M_0$ is
represented by the 4-tuple $(l_0,d_0,-1,-1)$, where $l_0$ and $d_0$
correspond to the label of the source node and its degree
respectively.  Then, if $i\geq 1$, $M_i$ is the 4-tuple
$(l_{i},d_i,p_i,q_i)$ where $l_i$ is the identifier of the occupied
node after the $i$th edge traversal, $d_i$ the degree of that node,
$p_i$ is the port by which the agent has left the node $l_{i-1}$ to
make the $i$th edge traversal, and $q_i$ the port by which the agent
entered node $l_i$ via its $i$th edge traversal.

%We assume that the adversary knows in advance the algorithm: it
%can thus choose the underlying graph and the starting node of the
%agent accordingly.

Without any constraint on the moves, the task of exploration in the above model is called \emph{unconstrained exploration}. However, in this paper, we study a variant introduced in \cite{DuncanKK06} and called \emph{distance-constrained exploration}. In this variant, the agent must deal with the additional constraint of always knowing a path (i.e., a sequence of already explored edges) of length at most $D$ from its current location to node $s$. The value $D$ is fixed to $(1+\alpha)r$, where $r$ and $\alpha$ respectively correspond to the eccentricity of node $s$ and an arbitrary positive real constant.  In \cite{DuncanKK06}, two scenarios are considered: one where both $r$ and $\alpha$ are initially known to the agent, and another where only $\alpha$ is known. In this paper, we adopt the more challenging scenario from a proof-of-impossibility perspective, where both $r$ and $\alpha$ are initially known to the agent, in order to give our result a greater impact.

An instance of the problem is defined as the tuple $(G,l_s,\alpha)$. We will denote by $\mathcal{I}(\alpha,r)$ the set of all instances $(G,l_s,\alpha)$ such that the eccentricity of the source node $s$ (of label $l_s$) in $G$ is $r$. 

An important notion involved in this paper is that of \emph{penalty}
incurred during exploration. First discussed by Pelc and Panaite in
\cite{PanaiteP99}, the \emph{penalty} of an exploration algorithm
running in $G$ is the number of edge traversals in excess of
$|E|$ made by the  mobile agent. In their seminal paper
\cite{PanaiteP99}, Panaite and Pelc gave an algorithm for solving
unconstrained exploration, which is guaranteed to work in every graph
$G$ with a (small) penalty of $\mathcal{O}(|V|)$. In the light of this result, an intriguing question naturally arises: \\

{\it Would it be possible to design a distance-constrained exploration algorithm working in every graph $G$ with a penalty of $\mathcal{O}(|V|)$?}

\subsection{Our Results}
\label{sec:ourresults}
In this paper, we provide a negative answer to the above question. This negative answer is strong as we show that for any positive real $\alpha$ and every integer $r\geq 6$, there is no algorithm that can solve distance-constrained exploration for every instance $(G=(V,E),l_s,\alpha)$ of $\mathcal{I}(\alpha,r)$ with a penalty of $\mathcal{O}(|V|)$. %Our result holds even if $r$, $\alpha$, $|V|$ and $|E|$ are known in advance to the agent (and thus to the algorithm).

Moreover, we observe that the impossibility result remains true for another variant of constrained exploration studied in the literature, known as \emph{fuel-constrained exploration}. In this variant, which was introduced by Betke et al. in \cite{BetkeRS95}, the context remains the same, except that now the agent must face a slightly different constraint: it has a fuel tank of limited size that can be replenished only at its starting node $s$. The size of the tank is $B=2(1+\alpha)r$, where $\alpha$ is any positive real constant and $r$ is the eccentricity of $s$. The tank imposes an important constraint, as it forces the agent to make at most {$\lfloor B\rfloor$} edge traversals before having to refuel at node $s$, otherwise the agent will be left with an empty tank and unable to move, preventing further exploration of the graph. The underlying graph and node $s$ are chosen so that $r\alpha\geq 1$ (if not, fuel-constrained exploration is not always feasible).

Through our two impossibility results, we show a fundamental
separation with the task of unconstrained exploration. We also solve
an open problem posed by Duncan et al. in \cite{DuncanKK06}, who asked
whether distance-constrained or fuel-constrained exploration could be
achieved with a penalty of $\mathcal{O}(|V|)$ in every
graph.\footnote{In \cite{DuncanKK06}, the authors also investigated a
third variant known as \emph{rope-constrained exploration}, but they
did not include it in their open problem. This point is further
detailed at the end of the related work section.}

%In the former variant, also introduced by Duncan \emph{et al.} in \cite{DuncanKK06}, the agent is tethered to its starting node $s$ by a rope that it unwinds by a length 1 with every forward edge traversal and rewinds by a length of 1 with every backward edge traversal. The length of the rope is $L=(1+\alpha)r$, where $r$ is the eccentricity of $s$ and $\alpha$ is any positive real such that $L\geq r+1$. Hence, at any point of an exploration, the difference between the number of forward edge traversals and the the number of a backward edge traversals can never exceed $L$.

\subsection{Related Work}

The problem of exploring unknown graphs has been the subject of
extensive research for many decades. It has been studied in scenarios
where exploration is conducted by a single agent and in scenarios
where multiple agents are involved. In what follows, we focus only on
work related to the former case (for the latter case, a good starting
point for the curious reader is the survey of Das \cite{Das19}).

In the context of theoretical computer science, the exploration problem was initially approached by investigating how to systematically find an exit in a maze (represented as a finite connected subgraph of the infinite 2-dimensional grid in which the edges are consistently labeled North, South, East, West). In \cite{Shannon51}, Shannon pioneered this field with the design of a simple finite-state automaton that can explore any small maze spanning a $5 \times 5$ grid. More than twenty years later, Budach \cite{Budach75} significantly broadened the scope of study by showing that no finite-state automaton can fully explore all mazes. Blum and Kozen showed in \cite{BlumK78} that this fundamental limitation can be overcome by allowing the automaton to use only two pebbles. The same authors naturally raise the question of whether one pebble could suffice. A negative answer to this question is provided by Hoffmann in \cite{Hoffmann81}.

In the following decades, much of the attention shifted to the exploration of arbitrary graphs, for which the solutions dedicated to mazes become obsolete. Studies on the exploration of arbitrary graphs can be broadly divided into two types: those assuming that nodes in the graph have unique labels recognizable by the agent, and those assuming anonymous nodes.

As for mazes, exploration of anonymous graphs has been investigated
under the scenario in which the agent is allowed to use pebbles
\cite{BenderFRSV02,ChalopinDK10,DieudonneP12,DudekJMW91}. In
\cite{DudekJMW91}, Dudek et al. showed that a robot provided with a
single movable pebble can explore any undirected anonymous graph using
$\mathcal{O}(|V|\cdot|E|)$ moves. If the pebble is not movable and can
only be used to mark the starting node of the agent, it is shown in
\cite{ChalopinDK10} that exploration is still possible using a number
of edge traversals polynomial in $|V|$. In \cite{DieudonneP12}, the
authors investigated the problem in the presence of identical
immovable pebbles, some of which are Byzantine (a Byzantine pebble can
be unpredictably visible or invisible to the agent whenever it visits
the node where the pebble is located) and designed an exploration
algorithm working in any undirected graph provided one pebble always remains
 fault-free. In the case where the graph is directed,
exploration becomes more challenging due to the impossibility of
backtracking and may even be sometimes infeasible if the graph is not
strongly connected. However, in the scenario where the agent evolves
in strongly connected directed graphs, Bender et al. proved that
$\Theta(\log \log |V|)$ movable pebbles are necessary and sufficient
to achieve a complete exploration.

When no marking of nodes is allowed, exploration of undirected anonymous graphs becomes obviously more difficult, and detecting the completion of the task cannot be always guaranteed without additional assumptions. If randomization is allowed, we know from \cite{AleliunasKLLR79} that a simple random walk of length $\mathcal{O}(\Delta^2|V|^3\log |V|)$ can cover, with high probability, all nodes of an undirected graph $G=(V,E)$ of maximal degree $\Delta$. Such a random walk can thus be used to fully explore a graph (and eventually stop), by initially providing the agent with an upper bound on $|V|$ and by making it go back and forth over each edge incident to a node of the random walk. If randomization is not allowed, all nodes of an undirected anonymous graph can still be visited, using universal exploration sequences (known as UXS) introduced in~\cite{KOUCKY2002717}. Formally, a sequence of integers $x_1,x_2,\dots,x_k$ is said to be a UXS for a class $\mathcal{G}$ of graphs, if it allows an agent, starting from any node of any graph $G\in\mathcal{G}$, to visit at least once every node of $G$ in $k+1$ steps as follows. In step $1$, the agent leaves its starting node $v_1$ by port $0$ and enters some node $v_2$. In step $2\leq i \leq k+1$, the agent leaves its current node $v_i$ by port $q=(p+x_{i-1})\mod {\tt deg}(v_i)$, where $p$ is the port by which it entered $v_i$ in step $i-1$, and enters some node $v_{i+1}$. In~\cite{Reingold08}, it is proven that for any positive integer $n$, a UXS of polynomial length in $n$, for the class of graphs of order at most $n$, can be computed deterministically in logarithmic space and in time polynomial  in $n$. As with random walks, this can be used to fully explore a graph, by initially providing the agent with an upper bound $n$ on $|V|$ and by instructing it to traverse each edge incident to a node visited by a UXS for the class of graphs of order at most $n$. Therefore, whether deterministically or not, it can be easily shown from \cite{AleliunasKLLR79,Reingold08} that an agent can explore any undirected anonymous graph of order at most $n$ even without the ability to mark nodes as long as it has a memory of size $\mathcal{O}(\log n)$: in \cite{FraigniaudIPPP05} Fraigniaud et al. gave a tight lower bound by showing that a memory of size $\Omega(\log n)$ is necessary.

The exploration of labeled graphs was a subject of research in \cite{AlbersH00,AwerbuchBRS99,AwerbuchK98,BetkeRS95,DengP99,DuncanKK06,PanaiteP99}, with a particular focus on minimizing the number of edge traversals required to accomplish the task.

In \cite{DengP99}, Deng and Papadimitriou showed an upper bound of
$d^{\mathcal{O}(d)}|E|$ moves to explore any labeled directed graph
that is strongly connected, where $d$, called the deficiency of the graph, is the minimum number of edges
that have to be added to make the graph Eulerian. This is subsequently
improved on by Albers and Henzinger in \cite{AlbersH00}, who gave a
sub-exponential algorithm requiring at most $d^{\mathcal{O}(\log
  d)}|E|$ moves and showed a matching lower bound of $d^{\Omega(\log
  d)}|E|$. For arbitrary labeled undirected graphs, the fastest
exploration algorithm to date is the one given by Panaite and Pelc in
\cite{PanaiteP99}, which allows an agent to traverse all edges using
at most $|E|+\mathcal{O}(|V|)$ moves. Their algorithm thus entails at
most a linear penalty in $|V|$, which is
significantly better than other classic strategies of exploration (for
instance, the standard Depth-First Search algorithm, which takes
$2|E|$ moves, has a worst-case penalty that is quadratic in
$|V|$). This upper bound of $\mathcal{O}(|V|)$ is asymptotically tight 
as the penalty can be in $\Omega(V)$, even in some Eulerian graphs. In particular, it is the case when an agent is required to explore a simple line made of an odd number $|V|\geq 3$ of nodes, by starting from the median node: the penalty is then necessarily at least $\frac{|V|-1}{2}$. 

The papers
\cite{AwerbuchBRS99,AwerbuchK98,BetkeRS95,DuncanKK06} are the closest
to our work. As mentioned in Section~\ref{sec:ourresults}, the problem
of fuel-constrained exploration of a labeled graph was introduced by
Betke et al. in~\cite{BetkeRS95}. In this paper, the authors gave
fuel-constrained algorithms using $\mathcal{O}(|E|+|V|)$ moves but for
some classes of graphs only. This line of study was then continued
in~\cite{AwerbuchBRS99} (resp. in~\cite{AwerbuchK98}) in which is
provided an $\mathcal{O}(|E|+|V|{\log^2|V|})$ algorithm (resp. an
$\mathcal{O}(|E|+|V|^{1+o(1)})$ algorithm) working in arbitrary
labeled graphs. Finally, Duncan et al. \cite{DuncanKK06} obtained
solutions requiring at most $\mathcal{O}(|E|+|V|)$ edge traversals for
solving fuel-constrained and distance-constrained
explorations. Precisely, the authors did not directly work on these
two variants, but on a third one called \emph{rope-constrained
exploration}, in which the agent is tethered to its starting node $s$
by a rope of length $L\in\Theta(r)$ that it unwinds by a length of $1$
with every forward edge traversal and rewinds by a length of $1$ with
every backward edge traversal. Hence, they designed an
$\mathcal{O}(|E|+|V|)$ algorithm for rope-constrained exploration
which can be transformed into algorithms solving the two other
constrained exploration problems with the same complexity. However,
these algorithms all carry a worst-case penalty proportional to
$|E|$ even when $|E|\in\Theta(|V|^2)$. As an open problem, Duncan et al. then asked whether we could
obtain algorithms for solving constrained explorations with a
worst-case penalty linear in $|V|$, as shown in \cite{PanaiteP99} when
there is no constraint on the moves. Precisely, Duncan et
al. explicitly raised the question for fuel-constrained and
distance-constrained explorations and not for rope-constrained
exploration. This is certainly due to the following straightforward
observation: at any point of a rope-constrained exploration, the
difference between the number of forward edge traversals and the
number of backward edge traversals can never exceed $L$. Thus, the
penalty of any rope-constrained exploration algorithm executed in any
graph $G=(V,E)$ is always at least $|E|$-$L$, which can be easily seen
as not belonging to $\mathcal{O}(|V|)$ in general, since $L$ belongs
to $\Theta(r)$ and thus to $\mathcal{O}(|V|)$.

\section{Preliminaries}
\label{sec:preli}
In this section, we introduce some basic definitions, conventions, operations, and results that will be used throughout the rest of the paper.

We call \emph{consistently labeled graph} a simple undirected labeled
graph in which nodes have pairwise distinct identifiers and in which
port numbers are fixed according to the rules given in
Section~\ref{sec:model}. We will sometimes use simple undirected
graphs with no identifiers and with no port numbers, especially when
describing intermediate steps of constructions: these graphs will be
called \emph{unlabeled graphs}. However, when the context makes it
clear or when the distinction between labeled and unlabeled graphs is
insignificant, we will simply use the term \emph{graph}.

A graph $G$ is said to be \emph{bipartite} if its nodes can be
partitioned into two sets $U$ and $U'$ such that no two nodes within
the same set are adjacent. If $G$ is also connected, such partition
$\{U,U'\}$ is unique and referred to as the \emph{bipartition} of $G$.
A graph is said to be \emph{regular} (resp. \emph{$k$-regular})
if each of its nodes has the same degree (resp. has degree
$k$). 

Below is a simple lemma about regular bipartite graphs. It follows from the fact that a $k$-regular bipartite graph of order $2n$, with $n\geq k\geq 1$, can be constructed by taking two sets of nodes $U=\{u_0,u_1,\ldots,u_{n-1}\}$ and $U'=\{u'_0,u'_1,\ldots,u'_{n-1}\}$ and then by adding edges so that for every $i$ and $j$ $\in [0..n-1]$, there is an edge between $u_i$ and $u'_j$ iff $j=(i+x)\mod n$ for some $x\in [0..k-1]$.
\begin{lemma}
\label{lem:biexist}
For every positive integer $k$, and for every integer $n\geq k$, there
exists a $k$-regular bipartite graph whose bipartition $\{U,U'\}$ satisfies $|U|=|U'|=n$.
\end{lemma}
\begin{comment}
\begin{proof}
  Let $k$ be any positive integer. We proceed by induction on
  $i=n-k$. For the base case $i=0$, we have $n=k$. Let $V$ be a set of
  $2n$ nodes. Split $V$ into two distinct subsets $U$ and $U'$, of $n$
  nodes each. Let $G$ be a graph with $V$ as set of nodes  where
  every vertex of $U$ is neighbor of every vertex in $U'$ i.e.,
  $G$ is a complete bipartite graph. By definition, $G$ is
  an $n$-regular bipartite graph, $|U|=|U'|=n$, and $\{U,U'\}$ is its
  unique bipartition since $G$ is connected (indeed, $n=k>0$).

  Consider now the case $i>0$. Then, $k=n-i$ and, by induction
  hypothesis, there exists a $(k+1)$-regular bipartite graph $B$ with
  a unique bipartition $\{U,U'\}$ satisfying $|U|=|U'|=n$. By
  Theorem~\ref{theo:perfect}, $B$ has a perfect matching $M$. By
  removing all edges of $M$ from $B$, we obtain a $k$-regular
  bipartite graph $B'$ such that $\{U,U'\}$ (with $|U|=|U'|=n$) is its
  unique bipartition since $B'$ is connected (indeed, $k>0$). This concludes the induction, and thus the proof of this lemma.
\end{proof}
\end{comment}

The impossibility proof described in Section~\ref{sec:imp} closely relies on a family of graphs $\mathcal{F}(\ell,w,r)$, where $\ell\geq 2$, $w\geq 4$ and $r\geq 6$ are integers. Precisely, $\mathcal{F}(\ell,w,r)$ is defined as the set of all consistently labeled graphs that can be constructed by successively applying the following five steps.

\begin{itemize}
\item{\bf Step~1.} Take a sequence of $\ell$ sets $V_1,V_2,V_3,\ldots,V_{\ell}$, each containing $w$ nodes. For every $1\leq i \leq \ell$, assign pairwise distinct labels to nodes of $V_i$ ranging from $w(i-1)+1$ to $wi$: these nodes will be called the \emph{nodes of level $i$}. Then, for every $1\leq i\leq \ell-1$, add edges from nodes of $V_i$ to nodes of $V_{i+1}$ so that the graph induced by $V_i\cup V_{i+1}$ is a $\lfloor\frac{w}{4}\rfloor$-regular bipartite graph denoted $B_i$ whose bipartition is $\{V_i,V_{i+1}\}$ (such a graph exists by Lemma~\ref{lem:biexist}).
  
\item{\bf Step~2.} This step is made of $\ell-1$ phases, and, in each phase $1\leq i \leq \ell-1$, we proceed as follows. Let $\beta=w\lfloor\frac{w}{4}\rfloor$ (which corresponds to the number of edges of the subgraph induced by $V_i\cup V_{i+1}$). Take a sequence of $\lfloor\frac{7\beta}{8}\rfloor$ new nodes $(g^i_1,g^i_2,g^i_3,\ldots,g^i_{\lfloor\frac{7\beta}{8}\rfloor})$ (by ``new nodes'', we precisely mean nodes that do not belong (yet) to the graph under construction). These nodes will be called \emph{gadgets}. For each $1\leq j \leq \lfloor\frac{7\beta}{8}\rfloor$, assign label $w\ell+ (i-1)\lfloor\frac{7\beta}{8}\rfloor+j$ to gadget $g^i_j$. Then, take a node $u$ of $V_i$ and a node $v$ of $V_{i+1}$ such that $u$ and $v$ are neighbors, remove the edge between $u$ and $v$, and finally add an edge between $u$ and $g^i_j$ as well as $v$ and $g^i_j$.
\item{\bf Step~3.} Take a pair of new nodes $(v_1,v_2)$. Assign label $0$ (resp. label $w\ell+ (\ell-1)\lfloor\frac{7\beta}{8}\rfloor+1$) to node $v_1$ (resp. node $v_2$). Then, add an edge between $v_1$ (resp. $v_2$) and each node of $V_1$ (resp. each gadget). Node $v_2$ will be called the \emph{critical node}.

\item{\bf Step~4.} Take a line $L$ made of $(r-3)$ new nodes. Add an edge between the critical node $v_2$ and one of the endpoints of $L$. Then, for each node $u$ of $L$, assign label $w\ell+ (\ell-1)\lfloor\frac{7\beta}{8}\rfloor+1+d$ where $d$ is the distance from node $v_2$ to $u$. The nodes of $L$ will form what will be called the \emph{tail} of the graph and the farthest endpoint of $L$ from $v_2$ will be called the \emph{tail tip}.
\item{\bf Step~5.} For every node $u$ belonging to the graph resulting from the successive application of steps 1 to 4, arbitrarily add pairwise distinct port numbers to each of its incident edges ranging from $0$ to ${\tt deg}(u)-1$.
\end{itemize}
Let us give some additional definitions. For every $1\leq i \leq \ell-1$, the layer $\mathcal{L}_G(i,i+1)$ of a graph $G$ resulting from the above construction is the set made of the following edges: those connecting a node of level $i$ with a node of level $i+1$, which will be called the \emph{green edges} of $\mathcal{L}_G(i,i+1)$, and those connecting a gadget of $(g^i_1,g^i_2,g^i_3,\ldots,g^i_{\lfloor\frac{7\beta}{8}\rfloor})$ with a node of level $i$ or $i+1$, which will be called the \emph{red edges} of $\mathcal{L}_G(i,i+1)$. When it is clear from the context, we will sometimes omit the reference to the underlying graph and use the notation $\mathcal{L}(i,i+1)$ instead of $\mathcal{L}_G(i,i+1)$. 

An example of a graph belonging to $\mathcal{F}(4,8,7)$ is depicted, with its key components, in Figure~\ref{fig:example}.

\begin{figure}[!t]
    \centering
    \includegraphics[width=0.7\textwidth]{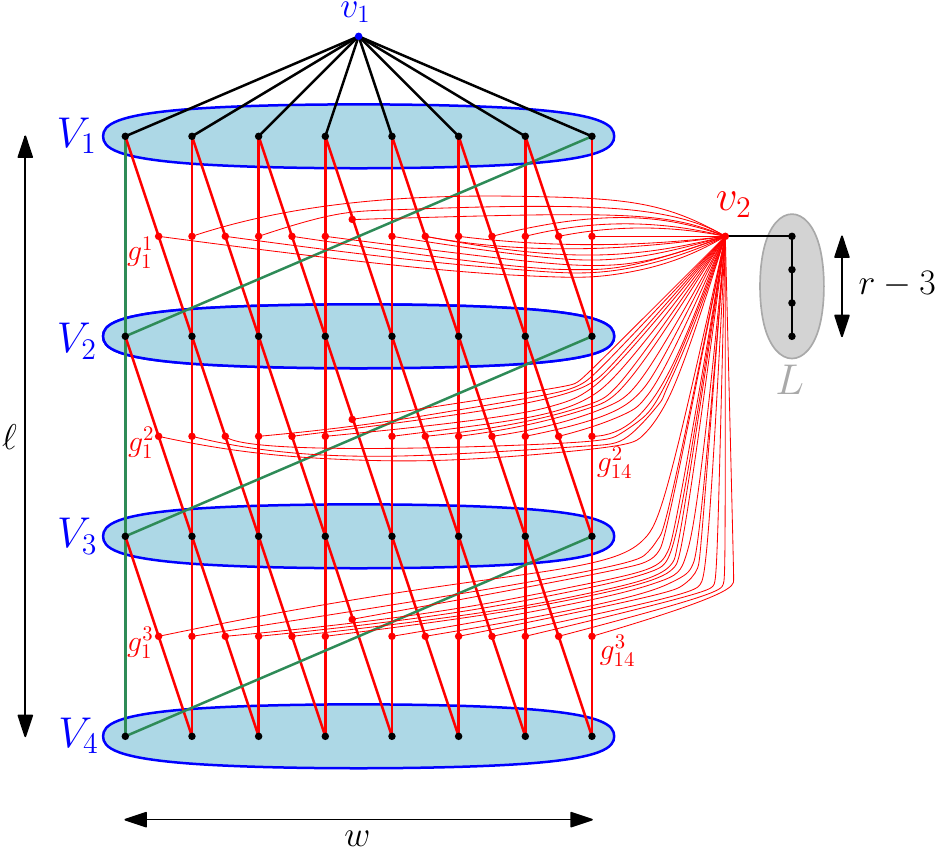}
    \caption{Example of a graph in $\mathcal{F}(4,8,7)$.}
    \label{fig:example}
\end{figure}

The next lemma is a direct consequence of our $5$-step construction.

\begin{lemma}
\label{lem:conse}
Let $\ell\geq 2$, $w\geq 4$ and $r\geq 6$ be integers. Let $\beta=w\lfloor\frac{w}{4}\rfloor$. For each graph $G$ of $\mathcal{F}(\ell,w,r)$, we have the following properties:
\begin{enumerate}
\item $G$ is consistently labeled. It contains $\ell-1$ layers and $\ell$ levels. 
\item The order of $G$ is $w\ell+ (\ell-1)\lfloor\frac{7\beta}{8}\rfloor+r-1$ and its number of gadgets is $(\ell-1)\lfloor\frac{7\beta}{8}\rfloor$.
\item For every $1\leq i \leq \ell-1$, the number of red edges (resp. green edges) in $\mathcal{L}_G(i,i+1)$ is $2\lfloor\frac{7\beta}{8}\rfloor$ (resp. $\lceil\frac{\beta}{8}\rceil$).
\item Two nodes that are incident to the same green edge in $G$ cannot be the neighbors of the same gadget.
%\item For every $1\leq i \leq \ell$, the number of nodes of level $i$ is $w$.
\end{enumerate} 
\end{lemma}

The lemma below addresses the eccentricity of the node labeled $0$ as
well as the length of some paths originating from it in graphs of the
family we have proposed. As mentioned in Section~\ref{sec:model}, a path is viewed as a sequence of edges (and not a sequence of nodes).

\begin{lemma}
\label{lem:ecc}
Let $\ell\geq 2$, $w\geq 16$ and $r\geq 6$ be integers. Let $G$ be a graph of $\mathcal{F}(\ell,w,r)$. We have the following two properties:
\begin{enumerate}
\item The eccentricity of the node of label $0$ in $G$ is $r$ and for every $2\leq i\leq \ell$, each node of level $i$ in $G$ is at distance at most $2$ from a gadget.
\item For every $1\leq i\leq \ell$, the length of a path from the node of label $0$ to any node of level $i$ in $G$ is at least $i$ if this path does not pass through the critical node of $G$.
\end{enumerate}
\end{lemma}

\begin{proof}
We first show that  each node of level $i \in [2..\ell]$ in $G$ is at distance at most $2$ from a gadget. To that goal, assume by contradiction that a node $p$ of some level $i \in [2..\ell]$ is at distance at least 3 from every gadget. In view of the construction of family $\mathcal{F}$, this implies that $p$ is the neighbor of $\lfloor\frac{w}{4}\rfloor$ nodes of level $i-1$ and each of them is incident to $\lfloor\frac{w}{4}\rfloor$ green edges of $\mathcal{L}_G(i-1,i)$. Moreover, by Lemma~\ref{lem:conse}, we know that the total number of green edges in $\mathcal{L}_G(i-1,i)$ is $\lceil\frac{\beta}{8}\rceil$ with $\beta=w\lfloor\frac{w}{4}\rfloor$. So, we have the following inequalities:

\begin{align*}
  \left(\left\lfloor\frac{w}{4}\right\rfloor\right)^2 & \leq \left\lceil\frac{\beta}{8}\right\rceil\\
%  \intertext{By definition, $\frac{w}{4}-1 < \lfloor\frac{w}{4}\rfloor \leq \frac{w}{4}$, so}
  \left(\frac{w}{4}-1\right)^2 & < \left\lceil\frac{\beta}{8}\right\rceil\\
\intertext{Since $\lceil\frac{n}{m}\rceil = \lfloor\frac{n-1}{m}\rfloor+1$, we have}
\left(\frac{w}{4}-1\right)^2 & < \left\lfloor\frac{\beta-1}{8}\right\rfloor+1\\
\intertext{Since $\beta=w\lfloor\frac{w}{4}\rfloor$, we have}
\left(\frac{w}{4}-1\right)^2 & < \left\lfloor\frac{w\lfloor\frac{w}{4}\rfloor-1}{8}\right\rfloor+1\\
%\intertext{By $\frac{w}{4}-1 < \lfloor\frac{w}{4}\rfloor \leq \frac{w}{4}$ again:}
\left(\frac{w}{4}-1\right)^2 & < \frac{w^2}{32}+1\\
%\frac{w^2}{16}-\frac{w}{2}+1 & < \frac{w^2}{32}+1\\
 -\frac{w^2}{32}+\frac{w}{2}  & > 0\\
\end{align*}
Now, $-\frac{w^2}{32}+\frac{w}{2} > 0$ if and only if $w \in (0,16)$. Since $w \geq 16$, this leads to a contradiction. Hence, each node of level $i \in [2..\ell]$ in $G$ is at distance at most $2$ from a gadget.
  
 We now show that the eccentricity of the node of label 0 is
 $r$. Recall that in the construction, $v_1$ is the node of label
 0. So, we have to show that the eccentricity of $v_1$ is $r$. Recall
 also that there are several types of nodes: namely, the critical node
 $v_2$, gadgets, nodes of the line $L$ and nodes in $V_i$ with $i \in
 [1..\ell]$.
  First, by construction, nodes of $V_1$ are at distance 1 from $v_1$.
  Then, all paths from the critical node $v_2$ to $v_1$ have length at
  least 3 since they pass through a gadget and a node of $V_1$. Now,
  there are nodes of $V_1$ linked to gadgets that are themselves
  connected to $v_2$. Consequently, $v_2$ is at distance exactly 3
  from $v_1$. Furthermore, being neighbor of $v_2$, each gadget is at
  distance at most 4 from $v_1$.
  Still by construction, all paths from nodes of $L$ to $v_1$ contain
  the critical node $v_2$ and are at distance at most $r-3$ from
  $v_2$; in particular the tail tip is exactly at distance $r-3$ from
  $v_2$. Overall, all nodes of $L$ are then at distance at most $r$
  from $v_1$; in particular, the tail tip is at distance exactly $r$
  from $v_1$.
   Finally, consider any node $u$ of some set $V_i$ with $i \in
   [2..\ell]$. We know that $u$ is at distance at most 2 from a gadget
   and so at distance at most 6 from $v_1$.
  As a result, it follows that the eccentricity of $v_1$ is $r$, which concludes the proof of the first property.

  For the second property, consider any path $P$ from $v_1$ to a node $u$ of some
  set $V_i$ with $i \in [1..\ell]$ that does not pass through $v_2$.
  By construction, $P$ should first reach $V_1$ and then
%% factorisé avec should
  traverse all levels from $1$ to
  $i-1$ before reaching $u$. Consequently, $P$ has length at least $i$, which means that the second property holds.
\end{proof}

%Given a node $v$ of a consistently labeled graph $G$ and a port $i$ at this node, ${\tt succ}(G,v,i)$ denotes the node that is reached by taking port $i$ from node $v$. Given a node $v'$ adjacent to $v$ in $G$, ${\tt port}(G,v,v')$ returns the port $i$ such that ${\tt succ}(v,i)=v'$.

In the impossibility proof of Section~\ref{sec:imp}, we will need to perform three operations on graphs. Let $G$ be a graph of $\mathcal{F}(\ell,w,r)$, with $\ell\geq 2$, $w\geq 4$ and $r\geq 6$. The first operation is ${\tt switch\mbox{-}ports}(G,v,p_1,p_2)$. If $v$ is a node of $G$ and $p_1$ and $p_2$ are distinct non-negative integers smaller than ${\tt deg}(v)$, this operation returns a graph $G'$ corresponding to $G$ but with the following modification: the edge having port $p_1$ (resp. $p_2$) at $v$ in $G$ now has port $p_2$ (resp. $p_1$) at $v$ in $G'$. Otherwise, it simply returns $G$ without any modification.

The second operation is ${\tt switch\mbox{-}edges}(G,v_1,v_2,p_1,p_2)$. Assume that for some $1\leq i\leq \ell-1$, $p_1$ and $p_2$ are the ports of green edges of $\mathcal{L}_G(i,i+1)$ at nodes $v_1$ and $v_2$ respectively, and these two nodes belong to a same level in $G$. Also assume that the node $v'_1$ (resp. $v'_2$) that is reached by taking port $p_1$ (resp. $p_2$) from $v_1$ (resp. $v_2$) is not a neighbor of $v_2$ (resp. $v_1$) and that $v_2$ and $v'_1$ (resp. $v_1$ and $v'_2$) are not the neighbors of a same gadget. In this case, this operation returns a graph $G'$ corresponding to $G$ but with the following modifications: the edge $e_1$ (resp. $e_2$), originally attached to port $p_1$ (resp. $p_2$) at node $v_1$ (resp. $v_2$), is detached from this node and reattached to node $v_2$ (resp. $v_1$) with port $p_2$ (resp. $p_1$) assigned at its newly attached node. Otherwise, it simply returns $G$ without any modification. An illustration of ${\tt switch\mbox{-}edges}$ is shown in Figure~\ref{fig:switchmove}.a.

Notice that the requirements for modifying $G$ using ${\tt switch\mbox{-}edges}(G,v_1,v_2,p_1,p_2)$ are essential. Indeed, the assumption that $v'_1$ (resp. $v'_2$) is not a neighbor of $v_2$ (resp. $v_1$) guarantees that we do not add an already existing edge: consequently, the graph obtained after the modification  is still simple with the same number of green edges and node's degrees are left unchanged. Moreover, the assumption that   $v_2$ and $v'_1$ (resp. $v_1$ and $v'_2$) are not the neighbors of a same gadget ensures that the fourth property of Lemma~\ref{lem:conse} remains true after the modification. Hence, if $G$ is a graph of $\mathcal{F}(\ell,w,r)$, then the graph returned by ${\tt switch\mbox{-}edges}(G,v_1,v_2,p_1,p_2)$ still belongs to $\mathcal{F}(\ell,w,r)$.

%% proposition Stéphane
  Finally, the third operation is ${\tt
    move\mbox{-}gadget}(G,e,g)$. Assume that for some $1\leq i\leq \ell-1$, $e$ is a green edge of
  $\mathcal{L}_G(i,i+1)$ and $g$ is a gadget incident to a red edge of $\mathcal{L}_G(i,i+1)$. Let $v_1$ and $v_2$ be the two nodes connected to $e$.
  %% plus utile
  %% and let $p_1$ (resp. $p_2$) be the port number
  %% at node $w_1$ (resp. $w_2$) of $e=\{w_1,w_2\}$.
  Let $u_i$ (resp. $u_{i+1}$) be the neighbor of level $i$ (resp. $i+1$)
  of the gadget $g$. In the following, we denote by $p_{(u,u')}$ the
  port number at node $u$ of the edge $\{u,u'\}$ in $G$ before applying
  the operation. 
The operation then returns a graph $G'$ obtained after applying on
  $G$ the following modifications in this order:
  \begin{enumerate}
  \item Remove the edges $e=\{v_1,v_2\}$, $\{u_i,g\}$, and $\{u_{i+1},g\}$.

  \item Add the edge $\{u_i,u_{i+1}\}$. Assign to this edge the port number $p_{(u_i,g)}$ (resp.  $p_{(u_{i+1},g)}$) at node $u_i$ (resp.
    $u_{i+1}$).

    \item Add the edge $\{v_1,g\}$. Assign to this edge the port number
      $p_{(v_1,v_2)}$ (resp. $p_{(g,u_i)}$) at node $v_1$ (resp. $g$).
      
    \item Add the edge $\{v_2,g\}$. Assign to this edge
      the port number $p_{(v_2,v_1)}$ (resp. $p_{(g,u_{i+1})}$) at node $v_2$ (resp. $g$).
  \end{enumerate}
  
In other cases, the operation simply returns $G$ without any modification. An illustration of ${\tt move\mbox{-}gadget}$ is shown in Figure~\ref{fig:switchmove}.b. Broadly speaking, when ${\tt move\mbox{-}gadget}$ causes a graph modification, a green edge is ``split'' into two red edges, while two red edges are ``merged'' into a single green edge.

\begin{figure}[!t]
\begin{center}
  \begin{minipage}[t]{0.49\linewidth}
    \centering
	\includegraphics[width=0.9\textwidth]{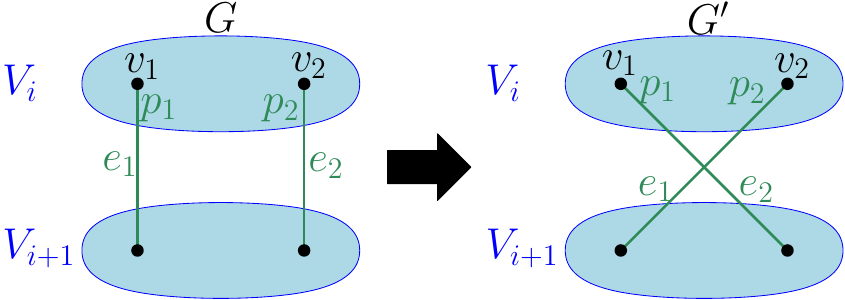}\\
    {\footnotesize ($a$) Operation ${\tt switch\mbox{-}edges}(G,v_1,v_2,p_1,p_2)$}
  \end{minipage}
  \begin{minipage}[t]{0.49\linewidth}
    \centering
	\includegraphics[width=0.9\textwidth]{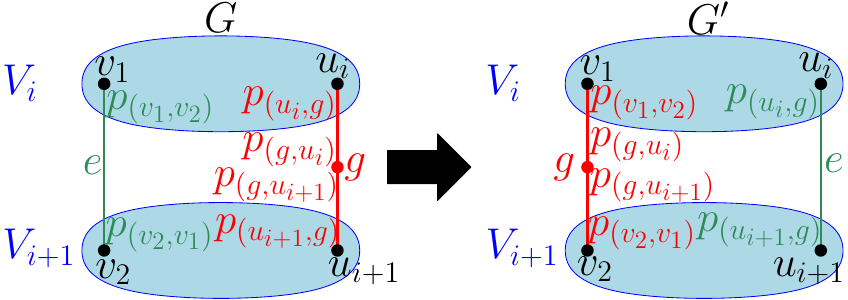}\\
    {\footnotesize ($b$) Operation ${\tt move\mbox{-}gadget}(G,e,g)$}
  \end{minipage}
\end{center}
 \caption{Illustrations of operations ${\tt switch\mbox{-}edges}$ and ${\tt move\mbox{-}gadget}$.}
\label{fig:switchmove}
\end{figure}

In view of the construction of $\mathcal{F}(\ell,w,r)$ and of the definitions of the three above operations, we have the following lemma.

\begin{lemma}\label{lem:modifGraph}
Let $\ell\geq 2$, $w\geq 4$ and $r\geq 6$ be integers. Let $G$ be a graph of $\mathcal{F}(\ell,w,r)$ and let $G'$ be the graph returned after the application of ${\tt switch\mbox{-}ports}(G,*,*,*)$, ${\tt switch\mbox{-}edges}(G,*,*,*,*)$ or ${\tt move\mbox{-}gadget}(G,*,*)$. $G'$ belongs to $\mathcal{F}(\ell,w,r)$.
% We have the following two properties.
%\begin{enumerate}
%\item $G'$ belongs to $\mathcal{F}(x,y,z)$.
%\item Let $\alpha$ be a positive real and let $A$ be an algorithm solving distance-constrained exploration, in $p$ edge traversals, for the instance $(G,s,\alpha)$, where $s$ is the node with label $0$ in $G$. Let $t\leq p$ be a positive integer. The memory of an agent applying $A$ from node $0$ in $G$ and  
%\end{enumerate}
\end{lemma}

\section{Impossibility Result for Distance-constrained Exploration}\label{sec:imp}

This section is dedicated to our proof showing the impossibility of solving distance-constrained exploration with a penalty that is always linear in the order of the graph. In what follows, we first give the intuitive ingredients that are behind our proof, and then the formal demonstration.

\subsection{Overview of the Proof}
\label{sec:int}

\subsubsection{In a Nutshell}
At a very high level, the proof can be viewed as composed of two main stages. The first stage consists in showing  that for every positive real $\alpha$, every integer $r \geq 6$, and every algorithm $A(\alpha, r)$ solving distance-constrained exploration for all instances of $\mathcal{I}(\alpha, r)$, there exists at least one instance $I=(G,0, \alpha)$ such that $G=(V,E) \in \mathcal{F}(\lfloor(1+\alpha)r\rfloor+1,w,r)$ and for which algorithm $A$ incurs a penalty of at least $\Omega(w^2)$ (for some arbitrarily large $w$). This result is achieved through a constructive proof in which an adversary, whose goal is to maximize the penalty, constructs ``on the fly'' the instance $I$ based on the step-by-step choices of algorithm $A$. If $w$ belonged to $\Theta(|V|)$, our work would essentially be completed here, as it would mean that algorithm $A$ would sometimes incur a penalty of $\Omega(|V|^2)$. Unfortunately, using Lemma~\ref{lem:conse}, we can show that $w$ is in $\mathcal{O}(\sqrt{|V|})$, which prevents us from concluding directly. This is where the second stage comes into the picture. By taking a closer look at our first stage, we will see that a penalty of at least $\Omega(w^2)$ is already incurred by algorithm $A$ with instance $I$ before the agent makes the $\mu$th edge traversal that leads it to explore a red edge of $G$ for the first time. In particular, before this traversal, in each layer, only green edges have been explored and no gadget nodes has been visited. However, the number of nodes that are not gadgets in $G$ is in $\Theta(r\times w)$ (cf. Lemma~\ref{lem:conse}). Therefore, if we can transform instance $I$ into an instance $I' = (G'=(V',E'),0, \alpha)$ of $\mathcal{I}(\alpha, r)$ by ``reducing'' the number of gadget nodes from $\Theta(r\times w^2)$ (cf. Lemma~\ref{lem:conse}) to $\Theta(w)$, while keeping the number of non-gadget nodes unchanged and ensuring that the execution of $A$ with instance $I$ is exactly the same as with $I'$ (from the point of view of the agent's memory) during the first $\mu$ edge traversals, we are done. Indeed, in this case, $|V'|\in \Theta(r\times w)$ and we can prove that algorithm $A$ necessarily incurs at least a penalty of order $\left(\frac{|V'|}{r}\right)^2$ with instance $I'$. Hence, it suffices to choose $w$ large enough so that $r^2$ is in $o(|V'|)$, to conclude that algorithm $A$ cannot guarantee a linear penalty in the order of the graph. Actually, it is precisely this transformation that is conducted in the second stage. Notice that proceeding in two stages, instead of a single one (in which the adversary would construct instance $I'$ directly), turns out to be more convenient because instance $I$, despite having too many nodes regarding the incurred penalty, is easier (from a pedagogical point of view) to construct dynamically than the instance $I'$ that has the appropriate number of nodes.

\subsubsection{Under the Hood}
Having described the general outline of our proof, let us now go deeper into the details of our two stages, by starting to give more intuitive explanations about the first one and, in particular, the behavior of the adversary. Here, the goal of the adversary is to construct an instance of $\mathcal{I}(\alpha, r)$ that creates a bad scenario for the algorithm. What could be a bad scenario?
% OLD
%% It could be one in which the agent frequently ``goes down'' to nodes at the penultimate level $\lfloor(1+\alpha)r\rfloor$ of a graph of $\mathcal{F}(\lfloor(1+\alpha)r\rfloor+1,w,r)$, with node $0$ as the source node, while there are still unexplored edges leading to the last level from these nodes and the agent has not yet explored any red edge.
%
% NEW
It could be one in which the agent, starting from node $0$ in a graph
of $\mathcal{F}(\lfloor(1+\alpha)r\rfloor+1,w,r)$, frequently ``goes
down'' until nodes at the penultimate level
$\lfloor(1+\alpha)r\rfloor$ while (1) there are still unexplored
edges incident to these nodes that lead to the last level and (2) no
red edge has been explored yet.
Indeed, what ensures that nodes at the last level are within distance $r$ from the source node are paths that pass exclusively through the critical node, which can only be accessed by traversing at least one red edge. Red edges and the critical node are essential components to make the eccentricity of the source node ``small''. As long as no red edge is explored, the critical node remains unexplored, and thus the agent lacks information about the shortest paths within the underlying graph. In particular, under these circumstances, if it reached a node of the last level $\lfloor(1+\alpha)r\rfloor+1$, the shortest path it would know to the source node would be of length at least  $\lfloor(1+\alpha)r\rfloor+1$ (cf. Lemma~\ref{lem:ecc}), which would violate the distance constraint. Hence, in such a scenario, each time the agent reaches a node of the penultimate level with possible edges leading to nodes of the last level, the algorithm cannot risk instructing the agent to traverse an unexplored edge. Actually, it has no other choice but to ask the agent to go back to the previous level by traversing an already explored edge, thereby incrementing the incurred penalty. And if this occurs at least $\Omega(w^2)$ times, the adversary successfully completes the first stage.

Hence, the question that naturally comes to mind is how to construct a graph implying such a bad scenario. As briefly mentioned earlier, the adversary can achieve this by dynamically modifying the underlying graph according to the choices made by the algorithm. More precisely, the adversary initially takes any instance $(G_0,0,\alpha)$ of $\mathcal{I}(\alpha, r)$ such that $G_0\in \mathcal{F}(\lfloor(1+\alpha)r\rfloor+1,w,r)$. Then, it emulates algorithm $A(\alpha, r)$ with this instance, but by paying particular attention to what happens before each edge traversal and making some possible modifications before each of them, and only if coherent and appropriate. By ``coherent'', we mean that a modification must not conflict with the prefix of the execution already made, and, in particular, with the agent's memory. For example, the adversary must not remove an already explored edge or change the degree of an already visited node. And by ``appropriate'', we mean that the modifications must be made only to enforce the bad scenario described above. Consequently, in the situation where the algorithm asks the agent to traverse an edge that has been explored before, the adversary does nothing, in order to preserve the coherence of the ongoing execution, while increasing the penalty. It also does nothing in the situation where the algorithm instructs the agent to take a port that makes it go from the source node (resp. a node of some level $i$) to a node of level $1$ (resp. $i+1$) that is incident to at least one unexplored green edge leading to a node of level $2$ (resp. $i+2$ if $i+2\leq \lfloor(1+\alpha)r\rfloor+1$), because this is essential for the bad scenario to unfold. In fact, if the adversary is lucky enough with its initial choice of $G_0$ to always get one of these situations (not necessarily the same) without intervening, during a sufficiently long prefix of the execution, we can inductively prove that the negative scenario occurs with $G_0$. Unfortunately, the adversary may not be so lucky. However, it can circumvent this, and enforce the bad scenario, by bringing modifications to the underlying graph (formally described in Algorithm~\ref{alg:alg1}) using clever combinations of the three operations ${\tt switch\mbox{-}ports}$, ${\tt move\mbox{-}gadget}$, and ${\tt switch\mbox{-}edges}$ defined in Section~\ref{sec:preli}. Each of these operations has a very specific role. Precisely, whenever possible, the adversary uses ${\tt switch\mbox{-}ports}$ to make the agent ``go down'', ${\tt move\mbox{-}gadget}$ to force the agent to traverse a green edge instead of a red one, and ${\tt switch\mbox{-}edges}$ to ensure that a node reached by the agent still has unexplored edges leading downward (which is important to subsequently continue the descent process and eventually enforce the agent to traverse an already explored edge). These operations are performed carefully to always ensure the coherence of the execution, while ensuring that the underlying graph remains in $\mathcal{F}(\lfloor(1+\alpha)r\rfloor+1,w,r)$ (cf. Lemma~\ref{lem:adv1}). One of the main challenges of the first stage consists in showing that the adversary can indeed intervene, with these operations, for a sufficiently long time to ultimately guarantee the desired penalty of $\Omega(w^2)$ before the first visit of a red edge and thus of a gadget. In fact, the general idea of the proof revolves around the following three key points (demonstrated through Lemmas~\ref{lem:adv2} and \ref{lem:penagadget}). First, the adversary can prevent the agent from visiting any red edge at least until half the green edges of some layer $\mathcal{L}(i,i+1)$ have been explored at some time $\lambda$. Second, as long as this event has not occurred, after each descending traversal of a green edge of $\mathcal{L}(i,i+1)$ (i.e., from level $i$ to level $i+1$), the agent will traverse an already explored edge before making another descending traversal of a green edge of $\mathcal{L}(i,i+1)$. This will happen either by the ``agent's own choice'' or because the adversary forces the agent to do so by bringing it to the penultimate level, as previously described. Hence, the incurred penalty is at least of the order of the number of descending traversals of green edges in $\mathcal{L}(i,i+1)$ made by time $\lambda$. Third and finally, as long as no red edge has been explored, the number of times the agent has made a descending traversal of a green edge in $\mathcal{L}(i,i+1)$ is always equal
%\footnote{The difference is at most $1$.}
to the number of times the agent has made an ascending traversal of a
green edge in $\mathcal{L}(i,i+1)$ (i.e., from level $i+1$ to level
$i$), give or take one. From these three key points and the fact that
there are $\Omega(w^2)$ green edges in $\mathcal{L}(i,i+1)$
(cf. Lemma~\ref{lem:conse}), it follows that the number of times the
agent makes a descending traversal of a green edge of this layer by
time $\lambda$ is in $\Omega(w^2)$, and so is the penalty.

To finish, let us now complete the overview with some intuitive explanations about the second stage. As stated before, the goal of this stage is to transform the instance $I=(G,0, \alpha)$ obtained in the first stage into an instance $I'=(G',0, \alpha)$ of $\mathcal{I}(\alpha, r)$ with only $\Theta(w)$ gadgets instead of $\Theta(r\times w^2)$. This reduction on the number of gadgets is done with merge operations that consist in ``merging'' some gadgets of $G$ into one gadget in $G'$ and setting the neighborhood of the new merged gadget as the union of the neighborhoods of the initial gadgets. An important property to check is that the source node of $G'$ (i.e., the node with label $0$) must have eccentricity $r$ for $I'$ to be a valid instance of $\mathcal{I}(\alpha, r)$. This is guaranteed by two facts. Firstly, merging nodes can only reduce distances between nodes and so the eccentricity of the source node in $G'$ is at most the eccentricity $r$ in $G$. Secondly, by the construction of $G$ (cf. Section~\ref{sec:preli}), all gadgets of $G$ are at distance at least $2$ from node $v_1$ (which plays here the role of the source node) and at distance at least $r-2$ from the tail tip $u$ of $L$. Since each path from $v_1$ to $u$ in $G$ contains a gadget, each path in $G'$ from $v_1$ to $u$ contains a merged gadget that is at distance at least $2$ from $v_1$ and at least $r-2$ from $u$. It follows that the eccentricity of the source node in $G'$ is also at least $r$.

Another important property to check, with merge operations, is that the new graph $G'$ must ``look'' the same from the point of view of non-gadget nodes contained in the layers of $G$, while being simple (i.e., without multiple edges). Indeed, this is also needed to ensure that $I'$ is a valid instance of $\mathcal{I}(\alpha, r)$, while guaranteeing that the execution of $A$ is the same in $I$ and $I'$ during the first $\mu$ edge traversals and thus the penalty up until the visit of the first gadget is the same in $I$ and $I'$. Hence, we cannot use trivial merge operations such as simply merging all gadgets of each layer into one gadget, or otherwise we immediately get a multigraph. We need something more subtle. Actually, we can show that the above property can be satisfied if the neighborhoods of the merged gadgets are disjoint, notably because this can avoid creating a multigraph, while keeping the same port numbers at non-gadget nodes for all edges connecting them to gadget nodes.

As a result, one of the key parts of the second stage is finding $\Theta(w)$ sets of gadgets that each have gadgets with disjoint neighborhoods. Intuitively, one can proceed in two steps. First, we consider each layer separately. Let us consider some layer $\mathcal{L}_G(i,i+1)$ of a graph $G$. One associates each gadget in this layer to an edge connecting its two neighbors in the layer. Finding sets of gadgets with disjoint neighborhoods can thus be viewed as finding sets of non-adjacent edges in this new graph. Using the fact that {the degree of nodes of level $i$ and $i+1$ in the graph induced by the edges of $\mathcal{L}_G(i,i+1)$ is in $\Theta(w)$, one can find $\Theta(w)$ such sets using a well-known bound for edge-coloring of graphs of bounded degree (cf. König's Theorem~\cite{konig1916}). Hence, one can merge the $\Theta(w^2)$ gadgets in the layer into $\Theta(w)$ gadgets, each corresponding to a color in the edge-coloring. With this first step, the number of gadgets is down from $\Theta(r\times w^2)$ to $\Theta(r\times w)$. The second step uses the fact that two gadgets in non-consecutive layers do not have common neighbors. Hence, one can partition the layers into odd and even layers such that gadgets in odd (resp. even) do not share neighbors. Then, one can take one merged gadget in each odd (resp. even) layer and merge all these gadgets into one. This permits to merge $\Theta(r)$ gadgets into one and by repeating this operation one can diminish the number of gadgets from $\Theta(r\times w)$ to $\Theta(w)$. This strategy, which is the crux of the second stage, is implemented in the proof of Theorem~\ref{theo:theo1} and constitutes the finishing touch to show our impossibility result.

\subsection{Formal Proof}
\label{sec:forproof}

Before going any further, we need to introduce some extra notations. Let $\alpha$ (resp. $r$) be any positive real (resp. any positive integer). Let $A(\alpha,r)$ be an algorithm solving distance-constrained exploration for all instances of $\mathcal{I}(\alpha,r)$ and let $(G,l_s,\alpha)$ be an instance of $\mathcal{I}(\alpha,r)$. We will denote by ${\tt cost}(A,(G,l_s,\alpha))$ the number of edge traversals prescribed by $A(\alpha,r)$ with instance $(G,l_s,\alpha)$. Moreover, for every integer $0\leq i \leq {\tt cost}(A,(G,l_s,\alpha))$, we will denote by ${\tt M}(A,(G,l_s,\alpha),i)$ (resp. ${\tt E}(A,(G,l_s,\alpha),i)$) the memory of the agent (resp. the set of edges that have been traversed by the agent) after the first $i$ edge traversals prescribed by $A(\alpha,r)$ with instance $(G,l_s,\alpha)$. We will also denote by $\overline{{\tt E}(A,(G,l_s,\alpha),i)}$ the set of edges of $G$ that are not in ${\tt E}(A,(G,l_s,\alpha),i)$. Finally, for every integer $1\leq i \leq {\tt cost}(A,(G,l_s,\alpha))$, we will denote by ${\tt node}(A,(G,l_s,\alpha),i)$ (resp. ${\tt edge}(A,(G,l_s,\alpha),i)$) the node of $G$ that is occupied after the $i$th edge traversal (resp. the edge of $G$ that is explored during the $i$th edge traversal) prescribed by $A(\alpha,r)$ with instance $(G,l_s,\alpha)$. By convention, the source node will be sometimes denoted by ${\tt node}(A,(G,l_s,\alpha),0)$.

Algorithm~\ref{alg:alg0} gives the pseudocode of function {\tt AdversaryBehavior}, which corresponds to the first stage outlined in Section~\ref{sec:int}. It takes as input the parameters $r$ and $\alpha$ of the distance-constrained exploration problem, an integer $w$ and an exploration algorithm $A$ solving the problem for all instances of $\mathcal{I}(\alpha,r)$. Under certain conditions (see Lemma~\ref{lem:penagadget}), the function returns a graph $G$ with the following two properties. The first property is that $(G,0,\alpha)$ is an instance of $\mathcal{I}(\alpha,r)$ and $G\in\mathcal{F}(\lfloor(1+\alpha)r\rfloor+1,w,r)$. The second property is that the penalty incurred by $A$ on instance $(G,0,\alpha)$, by the time of the first visit of a gadget, is in $\Omega(w^2)$ (i.e., the target penalty of the first stage). To achieve this, function {\tt AdversaryBehavior} relies on function {\tt GraphModification}, whose pseudocode is provided in Algorithm~\ref{alg:alg1}. Conceptually, {\tt AdversaryBehavior} proceeds in steps $0,1,2,3,$ etc. In step $0$ (cf. line~\ref{ligne:alg0:1} of Algorithm~\ref{alg:alg0}), we choose an arbitrary graph $G_0$ of $\mathcal{F}(\lfloor(1+\alpha)r\rfloor+1,w,r)$. In step $x\geq1$ (cf. lines~\ref{ligne:alg0:2} to~\ref{ligne:alg0:last} of Algorithm~\ref{alg:alg0}), we begin with a graph $G_{x-1}$ of $\mathcal{F}(\lfloor(1+\alpha)r\rfloor+1,w,r)$ for which the first $(x-1)$ edge traversals prescribed by $A$ from the node labeled $0$ are ``conformed'' with the bad scenario of the first stage. The goal is to ensure that this conformity extends to the first $x$ edge traversals as well. This is accomplished by deriving a graph $G_{x}$, which also belongs to $\mathcal{F}(\lfloor(1+\alpha)r\rfloor+1,w,r)$, from graph $G_{x-1}$. The derivation is handled by function {\tt GraphModification} that uses as inputs the graph $G_{x-1}$, the index $x-1$, the algorithm $A$ and the parameter $\alpha$. If $x<{\tt cost}(A,(G_x,0,\alpha))$, then there exists a $(x+1)$th edge traversal in the execution of $A$ from the node labeled $0$ in $G_x$, and thus we can continue the process with step $x+1$. Otherwise, the algorithm stops and returns graph $G_x$.

When we mentioned above that in step $x\geq1$, we aim at extending the bad scenario till the $x$th edge traversal included, we mean that {\tt GraphModification($G_{x-1},\alpha,A,x-1$)} returns a graph $G_{x}$ such that $(1)$ ${\tt M}(A,(G_{x-1},0,\alpha),x-1)={\tt M}(A,(G_{x},0,\alpha),x-1)$ (i.e., from the agent's perspective, there is no difference between the first $x-1$ edge traversals in $G_{x-1}$ and those in $G_{x}$) and $(2)$, if possible, the $x$th edge traversal in $G_{x}$ brings the agent closer to the penultimate level (in order to eventually force incrementing the penalty). The first point is particularly important for at least preserving the penalty already incurred, while the second is important for increasing it. However, while the first point will be always guaranteed by {\tt GraphModification} (cf. Lemma~\ref{lem:adv1}), the second is not always useful or even achievable. This may happen when the agent is at a node of the last level or has previously explored a red edge (in which cases the goal of the bad scenario should have been already met, as explained in Section~\ref{sec:int}) or when the agent is located at the source node (in which case the next move, if any, is in line with the bad scenario as it can only make the agent go down to a node of level $1$). This may also happen when the agent decides to recross an already explored edge as modifying the traversal would violate the agent's memory (but, this is not an issue, as it increases the penalty as intended). In all these situations (cf. the test of line~\ref{ligne:test} of Algorithm~\ref{alg:alg1}), no modifications will be made and $G_x=G_{x-1}$. (In the rest of the explanations, all line references are to Algorithm~\ref{alg:alg1}, which we will omit to mention for the sake of brevity).

%Denote by $G'_{x-1}$ the graph obtained just after the (possible) modifications made by lines~\ref{ligne:test:f}~to~\ref{ligne:modif:e:2} to $G_{x-1}$. Using function {\tt switch\mbox{-}ports}  (resp. {\tt move\mbox{-}gadget}), lines~\ref{ligne:modif:1} and~\ref{ligne:modif:e:1} (resp. lines~\ref{ligne:modif:2:deb} and~\ref{ligne:modif:e:2}) ensure that the $x$th edge traversal will make the agent move downward (resp. will correspond to a traversal of a green edge) if the test at line~\ref{ligne:test:f} (resp. line~\ref{ligne:test:k}) evaluates to true.

Assume that the test at line~\ref{ligne:test} evaluates to true during the execution of {\tt GraphModification($G_{x-1},\alpha,A,x-1$)} and denote by $G'_{x-1}$ the graph obtained right after the (possible) modifications made to $G_{x-1}$ by lines~\ref{ligne:test:f}~to~\ref{ligne:modif:e:2}. When evaluating the test at line~\ref{ligne:test}, the agent is thus located at some node $u$ of some level $i<\ell=\lfloor(1+\alpha)r\rfloor+1$ after the first $x-1$ edge traversals prescribed by $A$ in $G_{x-1}$. At this point, we may face several possible situations, which will be all addressed in the formal proof. But here, one is perhaps more important to pinpoint than the others. It is when the following two conditions are satisfied right before the $x$th edge traversal in $G_{x-1}$. The first condition, call it {$c1(u,i)$}, is that node $u$ is incident to at least one unexplored edge of $\mathcal{L}(i,i+1)$ and the second condition, call it {$c2(i)$}, is that there is at least one unexplored green edge in $\mathcal{L}(i,i+1)$. This situation is interesting because lines~\ref{ligne:test:f}~to~\ref{ligne:modif:e:2} guarantee that the $x$th edge traversal in $G'_{x-1}$ will correspond to a traversal of a green edge leading the agent from node $u$ (of level $i$) to a node $u'$ of level $i+1$, as intended in the bad scenario. Note that, if $i+1=\ell=\lfloor(1+\alpha)r\rfloor+1$ (i.e., the index of the last level in the graph) then algorithm $A$ violates the distance constraint (in view of the fact that no red edge has been yet visited), which is a contradiction. Hence, we have $i<\ell-1$ and, according to the bad scenario we aim at forcing the agent into, we want the $(x+1)$th move from node $u'$ (of level $i+1$) to lead the agent to level $i+2$ or to be a traversal of a previously explored edge. In the situation where $c1(u',i+1)$ and $c2(i+1)$ are also satisfied right before the $(x+1)$th edge traversal in $G'_{x-1}$, the test at line~\ref{ligne:test:u} evaluates to false and $G'_{x-1}=G_x$. This is a good situation for the bad scenario. Indeed, if the $(x+1)$th move in $G_x$ is not a traversal of a previously explored edge, then, in the next call to {\tt GraphModification}, the test at line~\ref{ligne:test} evaluates to true and lines~\ref{ligne:test:f}~to~\ref{ligne:modif:e:2} will ensure that the $(x+1)$th edge traversal in $G_x$ will correspond to a traversal of a green edge leading the agent to a node $u''$ of level $i+2$.

But what if we are not in such a good situation, where both $c1(u',i+1)$ and $c2(i+1)$ are satisfied right before the $(x+1)$th edge traversal in $G'_{x-1}$? In fact, if $c2(i+1)$ is not satisfied, then there is at least one layer for which at least half the green edges has been explored, which means that the target penalty has already been paid (as explained in Section~\ref{sec:int} and shown in the proof of Lemma~\ref{lem:penagadget}). Hence, the tricky case is when $c2(i+1)$ holds but $c1(u',i+1)$ does not, right before the $(x+1)$th edge traversal in $G'_{x-1}$. This is exactly where lines~\ref{llllll} to~\ref{ligne:modif:3} come into the picture (the test at line~\ref{ligne:test:u} then evaluates to true during the execution of {\tt GraphModification($G_{x-1},\alpha,A,x-1$)}). Using function {\tt switch\mbox{-}edges}, these lines attempt to reroute the $x$th edge traversal in $G'_{x-1}$ towards a node $v$, different from $u'$ but also of level $i+1$, such that $c1(v,i+1)$ holds right before the $(x+1)$th traversal in $G_x$. The feasibility of such rerouting will be guaranteed as long as some conditions are met, particularly that there is no layer for which at least half the green edges has been explored (cf. Lemma~\ref{lem:adv2}).

The switch performed at line~\ref{ligne:modif:3} involves two green edges. One of these edges is naturally $\{u,u'\}$. The other is a green edge (from the same layer as $\{u,u'\}$) resulting from the merge of two red edges achieved by {\tt move\mbox{-}gadget}, which is performed, right before the switch, at line~\ref{ligne:modif:4}. Proceeding in this way, rather than simply selecting an existing green edge, turns out to be essential to maximize the number of possible reroutings ensuring condition c1, and thus to get the desired penalty of $\Omega(w^2)$ before the first exploration of a red edge.

%Of course, $G_{t-1}$ may be already appropriate and thus we may simply have $G_{t-1}=G_{t}$.
\begin{comment}
Algorithm~\ref{alg:alg1} gives the pseudocode of function {\tt GraphModification} that takes as input a graph $G$, a positive real $\alpha$, a deterministic algorithm $A$ and an integer $t$. Provided that these parameters satisfy some conditions, given in Lemmas~\ref{lem:adv1} and~\ref{lem:adv2}, this function returns a graph $G^{new}$, whose properties are also given in the same lemmas. Roughly speaking, {\tt GraphModification} corresponds to a function that will be called by the adversary step after step in order to bring to the underlying graph the modifications alluded to in Section~\ref{sec:int}, which will force the bad scenario of the first stage to happen. The behavior of the adversary, based on this function, will be precisely formalized in the proof of Lemma~\ref{lem:penagadget}.   
\end{comment}

\begin{algorithm}[H]
\label{alg:alg0}
\caption{Function {\tt AdversaryBehavior($r,\alpha,A,w$)}}
\SetNoFillComment
\SetKwBlock{Repeat}{repeat}{}
\small
\tcc{Step 0}
$G :=$ any graph of $\mathcal{F}(\lfloor(1+\alpha)r\rfloor+1,w,r)$; $x := 1$\;\label{ligne:alg0:1}
\Repeat{
\tcc{Step x}
$G := $ {\tt GraphModification($G,\alpha,A,x-1$)}\;\label{ligne:alg0:2}

\eIf{$x < {\tt cost}(A,(G,0,\alpha))$\label{ligne:alg0:22}}
{
$x := x+1$\;\label{ligne:alg0:222}
}
{
\Return $G$\;\label{ligne:alg0:last}
}
}
\end{algorithm}

\begin{algorithm}[H]
\label{alg:alg1}
\caption{Function {\tt GraphModification($G,\alpha,A,t$)}}
%\SetKwFunction{DT}{Detach}
\small
$G' := G$; $r :=$ the eccentricity of the node having label $0$ in $G'$\;\label{ligne:1}
$u := {\tt node}(A,(G',0,\alpha),t)$; $e := {\tt edge}(A,(G',0,\alpha),t+1)$\;\label{ligne:2}
$p_u := $ the port of $e$ at node $u$; $\ell :=$ the number of levels in $G'$\;
\If{there is no red edge in ${\tt E}(A,(G',0,\alpha),t)$ \AND $e\notin{\tt E}(A,(G',0,\alpha),t)$ \AND $u$ is a node of some level $1\leq i< \ell$ in $G'$\label{ligne:test}}
{
%$i :=$ the level of node $u$ in $G'$\;
\If{$e\notin\mathcal{L}_{G'}(i,i+1)$ \AND $u$ is an endpoint of an edge $f\in\mathcal{L}_{G'}(i,i+1)\cap\overline{{\tt E}(A,(G',0,\alpha),t)}$\label{ligne:test:f}}
{
$p'_u :=$ the port of edge $f$ at node $u$; $G':= {\tt switch\mbox{-}ports}(G',u,p_u,p'_u)$\;\label{ligne:modif:1}
 $e := {\tt edge}(A,(G',0,\alpha),t+1)$\;}\label{ligne:modif:e:1}\label{ligne:7}

\If{$e$ is a red edge of some layer $\mathcal{L}$ \AND there is a green edge $k\in\mathcal{L}\cap\overline{{\tt E}(A,(G',0,\alpha),t)}$\label{ligne:test:k}}
{
$u' := {\tt node}(A,(G',0,\alpha),t+1)$\;\label{ligne:modif:2:deb}
%$v :=$ the node of level $i$ that is connected to $e$; $v' :=$ the node of level $i+1$ that is connected to $e$\;
$G' := {\tt move\mbox{-}gadget}(G',k,u')$; $e := {\tt edge}(A,(G',0,\alpha),t+1)$\;\label{ligne:modif:2}\label{ligne:modif:e:2}
}

$u' := {\tt node}(A,(G',0,\alpha),t+1)$\;\label{ligne:modif:2:fin} \label{ligne:11}
 
\If{$i<\ell-1$ \AND $u'$ is a node of level $i+1$ in $G'$ \AND $u'$ is not an endpoint of an edge $\in\mathcal{L}_{G'}(i+1,i+2)\cap\overline{{\tt E}(A,(G',0,\alpha),t)}$\label{ligne:test:u}}{

$D :=$ the set made of the two endpoints of $e$, plus all gadgets in $G'$ that are incident to a red edge of $\mathcal{L}_{G'}(i,i+1)$ and that are adjacent to at least one of the endpoints of $e$\;\label{llllll}
                
$N_i :=$ the set of all nodes of level $i$ in $G'$ that are not adjacent to a node of $D$\;

$N_{i+1} :=$ the set of all nodes of level $i+1$ in $G'$ that are not adjacent to a node of $D$ and that are incident to at least one edge of $\mathcal{L}_{G'}(i+1,i+2)\cap\overline{{\tt E}(A,(G',0,\alpha),t)}$\;

\If{there exist a green edge $h\in(\mathcal{L}_{G'}(i,i+1)\cap\overline{{\tt E}(A,(G',0,\alpha),t)})\setminus \{e\}$ and a gadget $g$ adjacent to a node of $N_i$ and to a node of $N_{i+1}$\label{ligne:test:h}}
{
$G' := {\tt move\mbox{-}gadget}(G',h,g)$\;\label{ligne:modif:4} 
$k :=$ the green edge of $G'$ created by the previous call to ${\tt move\mbox{-}gadget}$\;\label{ligne:yoyo} \label{ligne:18}
$v :=$ the node of level $i+1$ that is incident to edge $k$\;
$p_{v} :=$ the port of edge $k$ at node $v$\;
%$u' := {\tt node}(A,(G',0,\alpha),t+1)$; $e := {\tt edge}(A,(G',0,\alpha),t+1)$\;\label{ligne:modif:e:last}
$p_{u'} :=$ the port at node $u'$ of $e$\;
$G' := {\tt switch\mbox{-}edges}(G',u',v,p_{u'},p_v)$\;\label{ligne:modif:3}
  }
	}
}
\Return $G'$\;\label{ligne:fin}\label{ligne:23}
\end{algorithm}

For the next two lemmas, we first observe that for every positive real
$\alpha$, every integer $r \geq 6$, every integer $w\geq 4$, and every
$G \in \mathcal{F}(\lfloor(1+\alpha)r\rfloor+1,w,r)$, we have
$(G,0,\alpha) \in \mathcal{I}(\alpha,r)$, by Lemma~\ref{lem:ecc}. 

\begin{lemma}
  \label{lem:adv1}
  Let $\alpha$ be any positive real. Let $w \geq 16$ and $r\geq 6$ be
  two integers.
  Let $A(\alpha,r)$ be an algorithm solving the
  distance-constrained exploration problem for every instance of
  $\mathcal{I}(\alpha,r)$.
  For every $G\in\mathcal{F}(\lfloor(1+\alpha)r\rfloor+1,w,r)$ and every integer $t$ such that $0\leq t < {\tt
    cost}(A,(G,0,\alpha))$, we let $G^{new} = $ {\tt GraphModification}($G,\alpha,A,t$) and we have the following
  three properties:
\begin{enumerate}

\item $G^{new}\in\mathcal{F}(\lfloor(1+\alpha)r\rfloor+1,w,r)$ and
  $(G^{new},0,\alpha)$ is an instance of $\mathcal{I}(\alpha,r)$;
  
\item During the execution of {\tt GraphModification}($G,\alpha,A,t$), we have 
  ${\tt M}(A,(G,0,\alpha),t)={\tt M}(A,(G',0,\alpha),t)$ after each
  modification of the variable $G'$, consequently ${\tt M}(A,(G,0,\alpha),t)={\tt M}(A,(G^{new},0,\alpha),t)$;

\item If ${\tt node}(A,(G,0,\alpha),t)$ is a node of some level
  $1\leq i\leq\lfloor(1+\alpha)r\rfloor+1$ in $G$, there is no red edge in ${\tt E}(A,(G,0,\alpha),t)$, and, for all $1\leq j \leq
  \lfloor(1+\alpha)r\rfloor$, there is a green edge in
  $\mathcal{L}_{G}(j,j+1)\cap\overline{{\tt E}(A,(G,0,\alpha),t)}$,
  then ${\tt node}(A,(G^{new},0,\alpha),t+1)$ is not a gadget.
\end{enumerate}
\end{lemma}

\begin{proof}
Throughout this proof, each time we refer to a specific line, it is one of Algorithm~\ref{alg:alg1}, and thus we
omit to mention it, in order to facilitate the reading.

  Consider the first property. By hypothesis,
  $G\in\mathcal{F}(\lfloor(1+\alpha)r\rfloor+1,w,r)$. In {\tt
    GraphModification}, $G'$ is initialized to $G$ (see
  line~\ref{ligne:1}) and during the execution of {\tt
    GraphModification}, $G'$ may only be modified using functions {\tt
    switch\mbox{-}ports} (line~\ref{ligne:modif:1}), {\tt
    move\mbox{-}gadget} (lines~\ref{ligne:modif:2}
  and~\ref{ligne:modif:4}), and {\tt switch\mbox{-}edges}
  (line~\ref{ligne:modif:3}). So, by Lemma~\ref{lem:modifGraph}, at the beginning and after each modification, the graph $G'$ always belongs to $\mathcal{F}(\lfloor(1+\alpha)r\rfloor+1,w,r)$ and so does the graph $G^{new}$ returned by {\tt GraphModification}($G,\alpha,A,t$). Moreover, $(G^{new},0,\alpha)$ is an instance
  of $\mathcal{I}(\alpha,r)$, by Lemma~\ref{lem:ecc}. This concludes
  the proof of the first property.

  Consider the second property and let $G'(j)$ be the value of $G'$, in the call {\tt
    GraphModification}($G,\alpha,A,t$),
  right before executing line $j$.
  $G'$ is initialized at line~\ref{ligne:1} and may be modified
  at line~\ref{ligne:modif:1} using functions {\tt
    switch\mbox{-}ports}, at
  lines~\ref{ligne:modif:2},~\ref{ligne:modif:4} using {\tt
    move\mbox{-}gadget}, and/or at line~\ref{ligne:modif:3} using {\tt
    switch\mbox{-}edges}. 
  So, we have to show that for every $j \in\{\ref{ligne:2},
  \ref{ligne:7}, \ref{ligne:11}, \ref{ligne:18},\ref{ligne:23}\}$, if
  line $j-1$ is executed, then ${\tt M}(A,(G,0,\alpha),t)={\tt
    M}(A,(G'(j),0,\alpha),t)$.
First, since $G'$ is initialized to $G$ at line \ref{ligne:1}, ${\tt M}(A,(G,0,\alpha),t)={\tt
    M}(A,(G'(2),0,\alpha),t)$ trivially holds.
The other cases may only happen if the test of line \ref{ligne:test}
is true. So from now on, assume this is the case.  After this test,
$G'$ may be modified using functions {\tt switch\mbox{-}ports}
(line~\ref{ligne:modif:1}), {\tt move\mbox{-}gadget}
(lines~\ref{ligne:modif:2},~\ref{ligne:modif:4}), and/or {\tt
  switch\mbox{-}edges} (line~\ref{ligne:modif:3}). Note that these
functions do not modify the number of nodes, the number of edges, the
nodes' degrees, or the nodes' labeling. Hence, to establish that the
property holds in all remaining cases, it is sufficient to show that
each call to these functions does not alter edges in ${\tt
  E}(A,(G,0,\alpha),t)$. To that goal and by their definitions, we just
need to show that these calls do not involve already explored edges.
\begin{itemize}
  \item Consider first the call to ${\tt
    switch\mbox{-}ports}(G',u,p_u,p'_u)$ at
    line~\ref{ligne:modif:1}. Here, $p_u$ (resp. $p'_u$) is a port of
    an edge $e$ (resp. $f$) and from the test at line \ref{ligne:test}
    (resp. line \ref{ligne:test:f}) we know that $e$ (resp. $f$) is an
    edge of $\overline{{\tt E}(A,(G,0,\alpha),t)}$. Thus, the call to
    ${\tt switch\mbox{-}ports}(G',u,p_u,p'_u)$ does not alter edges in
    ${\tt E}(A,(G,0,\alpha),t)$. Consequently, ${\tt
      E}(A,(G'(\ref{ligne:7}),0,\alpha),t) = {\tt
      E}(A,(G,0,\alpha),t)$, which implies ${\tt
      M}(A,(G,0,\alpha),t)={\tt M}(A,(G'(\ref{ligne:7}),0,\alpha),t)$.

\item Consider now the call to ${\tt move\mbox{-}gadget}(G',k,u')$ at
  line~\ref{ligne:modif:2}. This call involves two red edges (incident
  to a gadget $u'$) and a green edge $k$.  Now, from the previous
  case, we can deduce that ${\tt E}(A,(G'(\ref{ligne:test:k}),0,\alpha),t) = {\tt E}(A,(G),0,\alpha),t)$. So, the two
  involved red edges do not belong to ${\tt E}(A,(G'(\ref{ligne:test:k}),0,\alpha),t) = {\tt E}(A,(G),0,\alpha),t)$, by
  hypothesis and lines \ref{ligne:test:k} and \ref{ligne:modif:2:deb}, and the green edge $k$ belongs to $\overline{{\tt
      E}(A,(G'(\ref{ligne:test:k}),0,\alpha),t)}$ (see line
  \ref{ligne:test:k}) and so not to ${\tt
    E}(A,(G,0,\alpha),t)$. Hence, the call to ${\tt
    move\mbox{-}gadget}(G',k,u')$ does not alter edges in ${\tt
    E}(A,(G,0,\alpha),t)$. Consequently, ${\tt
    E}(A,(G'(\ref{ligne:11}),0,\alpha),t) = {\tt
    E}(A,(G,0,\alpha),t)$, which implies ${\tt
    M}(A,(G,0,\alpha),t)={\tt M}(A,(G'(\ref{ligne:11}),0,\alpha),t)$.

\item Based on line \ref{ligne:test:h}, the call
  to ${\tt move\mbox{-}gadget}(G',h,g)$ at line \ref{ligne:modif:4}
  can be handled similarly to the previous case and so we have ${\tt
    E}(A,(G'(\ref{ligne:18}),0,\alpha),t) = {\tt
    E}(A,(G,0,\alpha),t)$, which implies ${\tt
    M}(A,(G,0,\alpha),t)={\tt M}(A,(G'(\ref{ligne:18}),0,\alpha),t)$.

\item Finally, let us focus on the call to ${\tt
  switch\mbox{-}edges}(G',u',v,p_{u'},p_v)$ at line
  \ref{ligne:modif:3}. Here, $p_v$ (resp. $p_{u'}$) is the port of an
  edge $k=\{v,v'\}$ (resp. edge $e=\{u,u'\}$) in $G'(\ref{ligne:modif:3})$. Edges $e$ and $k$ are the only edges of $G'(\ref{ligne:modif:3})$ impacted by this call to ${\tt
  switch\mbox{-}edges}$. From the previous
  cases, we know that we have ${\tt E}(A,(G,0,\alpha),t) = {\tt
    E}(A,(G'(\ref{ligne:modif:4}),0,\alpha),t)$. Moreover, $v$ and $v'$
  are linked to the gadget $g$ in $G'(\ref{ligne:modif:4})$. Thus,
  $G'(\ref{ligne:modif:4})$ does not contain $k$, by Lemma
  \ref{lem:conse}. Consequently, $k \notin {\tt
    E}(A,(G'(\ref{ligne:18}),0,\alpha),t)$, which in turn implies
  that $k \notin {\tt E}(A,(G,0,\alpha),t)$.
  
   Let us focus now on edge $e$. We know that to execute
   line \ref{ligne:modif:3}, the test at line \ref{ligne:test} shall
   be true, so initially $e \notin {\tt E}(A,(G,0,\alpha),t)$. Then,
   if $e$ is modified at line~\ref{ligne:modif:e:1}, $e$ is set to an
   edge $f \notin {\tt E}(A,(G,0,\alpha),t)$; see the first case. If
   $e$ is modified at line~\ref{ligne:modif:e:2}, then using a reasoning similar to the one used for edge $k$, we can conclude that $e \notin {\tt
     E}(A,(G,0,\alpha),t)$. Then, $e$ is no more modified. So, $e \notin {\tt
     E}(A,(G,0,\alpha),t)$ right before executing line \ref{ligne:modif:3}.

   To
   conclude, $k \notin {\tt E}(A,(G,0,\alpha),t)$ and $e \notin {\tt
     E}(A,(G,0,\alpha),t)$ when ${\tt
     switch\mbox{-}edges}(G',u',v,p_{u'},p_v)$ is performed and so
   this call does not alter edges in ${\tt E}(A,(G,0,\alpha),t)$. 
   Consequently, ${\tt M}(A,(G,0,\alpha),t)={\tt
     M}(A,(G'(\ref{ligne:23}),0,\alpha),t)$ and we are done with the
  proof of the second property.
\end{itemize}

  Finally, we consider the third property. Let $u = {\tt
    node}(A,(G,0,\alpha),t)$. Assume first that $u$ is a node of level
  $\lfloor(1+\alpha)r\rfloor+1$. It follows that the critical
  node has been visited by the time the $t$th edge traversal is completed in the execution of $A$ with instance $(G,0,\alpha)$, because otherwise, $A$ violates the distance constraint with this instance by Lemma~\ref{lem:ecc}, which is a contradiction. However, if the critical node has been visited by the time the $t$th edge traversal is completed, it means that ${\tt E}(A,(G,0,\alpha),t)$ contains at least one red edge, which is impossible by hypothesis.
  Hence, to show the third property, we can assume that $u$ is a node of some level $1\leq i <
  \lfloor(1+\alpha)r\rfloor+1$. Recall that the labels of the nodes of level $i$ in $G'(j)$, if line~j is executed, are the same as those of level $i$ in $G$ (as $G'(j)\in\mathcal{F}(\lfloor(1+\alpha)r\rfloor+1,16k,r)$ as shown in the proof of the first property. Thus, since by the second property of the lemma, already proven above, we have ${\tt M}(A,(G,0,\alpha),t)={\tt M}(A,(G'(j),0,\alpha),t)$, we know that ${\tt node}(A,(G'(j),0,\alpha),t)$ is also a node of level $i$, if line~$j$ is executed.

Since there is no red edge in ${\tt E}(A,(G,0,\alpha),t)$ (by hypothesis), if the test at line
  \ref{ligne:test} is false, then the next edge to be traversed has
  been already explored and so is not a red edge, still by
  hypothesis. Consequently, ${\tt node}(A,(G^{new},0,\alpha),t+1)$ is not a
  gadget.  Assume otherwise that the test at line \ref{ligne:test} is
  true. By hypothesis, there is a green edge in
  $\mathcal{L}_{G}(i,i+1)\cap\overline{{\tt
      E}(A,(G,0,\alpha),t)}$ and thus in $\mathcal{L}_{G'(\ref{ligne:test:f})}(i,i+1)\cap\overline{{\tt
      E}(A,(G'(\ref{ligne:test:f}),0,\alpha),t)}$. So, the test at line \ref{ligne:test:f}, lines \ref{ligne:modif:1}-\ref{ligne:modif:e:1}, and the fact that ${\tt node}(A,(G'(\ref{ligne:modif:e:1}),0,\alpha),t)$ is a node of level $i$, ensure that $e
  \in \mathcal{L}_{G'(\ref{ligne:test:k})}(i,i+1)$ when the test of line \ref{ligne:test:k} is
  performed. Moreover, if $G'$ is modified at line
  \ref{ligne:modif:1}, then this modification does not alter its
  topology (only some port numbers are modified). So, there still is a
  green edge in $\mathcal{L}_{G'(\ref{ligne:test:k})}(i,i+1)\cap\overline{{\tt
      E}(A,(G'(\ref{ligne:test:k}),0,\alpha),t)}$ when the test of line \ref{ligne:test:k}
  is performed. Consequently, the test at line \ref{ligne:test:k},
  lines \ref{ligne:modif:2:deb}-\ref{ligne:modif:2} and the fact that ${\tt node}(A,(G'(\ref{ligne:11}),0,\alpha),t)$ is a node of level $i$, ensure that ${\tt edge}(A,(G'(\ref{ligne:11}),0,\alpha),t+1)$ is a green edge of $\mathcal{L}_{G'(\ref{ligne:11})}(i,i+1)$  and so ${\tt node}(A,(G'(\ref{ligne:11}),0,\alpha),t+1)$ is not gadget. Thus, if $G'$
  is no more modified after line~\ref{ligne:11}, we are done as $G'(\ref{ligne:11})=G^{new}$. 

Otherwise, the condition at line~\ref{ligne:test:h} evaluates to true and line~\ref{ligne:modif:4} to~\ref{ligne:modif:3} are executed. In view of the above explanations, note that ${\tt edge}(A,(G'(\ref{ligne:11}),0,\alpha),t+1)=\{u,u'\}$, where $u$ and $u'$ are respectively of level $i$ and $i+1$.

With the {\tt move\mbox{-}gadget} operation of line~\ref{ligne:modif:4}, a gadget $g$, which is adjacent to a node $v\in N_i$ and a node $v'\in N_{i+1}$, and the red edges incident to $g$ are replaced by an edge $k=\{v,v'\}$, while a green edge $\{x,x'\}$ of $(\mathcal{L}_{G'(\ref{ligne:test:u})}(i,i+1)\cap\overline{{\tt E}(A,(G'(\ref{ligne:test:u}),0,\alpha),t)})\setminus \{e\}$ (with $e={\tt edge}(A,(G'(\ref{ligne:test:u}),0,\alpha),t+1)=\{u,u'\}$) is replaced by a gadget $g'$ and two red edges $\{x,g'\}$ and $\{x',g'\}$. Observe that since $v\in N_i$ and $v'\in N_{i+1}$, there is no edge in $G'(\ref{ligne:test:u})$ linking $v$ or $v'$ to an endpoint of $e={\tt edge}(A,(G'_{\ref{ligne:test:u}},0,\alpha),t+1)=\{u,u'\}$ or a gadget adjacent to an endpoint of $e$. In view of the {\tt move\mbox{-}gadget} operation at line~\ref{ligne:modif:4} and the fact that ${\tt M}(A,(G'({\ref{ligne:test:u}}),0,\alpha),t)={\tt M}(A,(G'({\ref{ligne:modif:4}}),0,\alpha),t)={\tt M}(A,(G'({\ref{ligne:modif:3}}),0,\alpha),t)$ (cf. the second property proven earlier), this remains true just before executing the {\tt switch\mbox{-}edge} operation of line~\ref{ligne:modif:3}, i.e., there is no edge in $G'(\ref{ligne:modif:3})$ linking $v$ or $v'$ to an endpoint of $e={\tt edge}(A,(G'_{\ref{ligne:modif:3}},0,\alpha),t+1)=\{u,u'\}$ or a gadget adjacent to an endpoint of $e$. Hence, by the definition of the {\tt switch\mbox{-}edge} operation, when executing line~\ref{ligne:modif:3}, edge $\{u,u'\}$ and $\{v,v'\}$ are replaced by edges $\{u,v'\}$ and $\{v,u'\}$. In particular, the port at node $u$ of edge $\{u,v'\}$ is the same as the port that had edge $\{u,u'\}$ at node $u$. Thus, since $u={\tt node}(A,(G'(\ref{ligne:test:u}),0,\alpha),t)$ and ${\tt M}(A,(G'({\ref{ligne:test:u}}),0,\alpha),t)={\tt M}(A,(G'({\ref{ligne:modif:3}}),0,\alpha),t)={\tt M}(A,(G'({\ref{ligne:fin}}),0,\alpha),t)$, it follows that $v'={\tt node}(A,(G'(\ref{ligne:fin}),0,\alpha),t+1)$. Note that, since node $v'$ belongs to $N_{i+1}$, we know that it is a node of level $i+1$ in $G'(\ref{ligne:test:u})$. However, the set of labels of the nodes at level $i+1$ in
  $G'(\ref{ligne:test:u})$ is identical to the set of labels of the nodes at level
  $i+1$ in $G'(\ref{ligne:fin})$ because $G'(\ref{ligne:test:u})$ and $G'(\ref{ligne:fin})$ both belong to $\mathcal{F}(\lfloor(1+\alpha)r\rfloor+1,16k,r)$ (cf. the proof of the first property). Therefore, ${\tt node}(A,(G'(\ref{ligne:fin}),0,\alpha),t+1)={\tt node}(A,(G^{new},0,\alpha),t+1)$ is a node of level $i+1$, which concludes the proof of the third property, and thus the proof of this lemma.
\end{proof}

\begin{lemma}
\label{lem:adv2}
Let $\alpha$ be any positive real. Let $k>0$ and $r\geq 6$ be two integers.
Let $A(\alpha,r)$ be an algorithm solving the distance-constrained exploration problem for every instance of $\mathcal{I}(\alpha,r)$.
For every $G\in\mathcal{F}(\lfloor(1+\alpha)r\rfloor+1,16k,r)$ and every integer $t$ such that $1\leq t < {\tt cost}(A,(G,0,\alpha))$, if the following three conditions are satisfied: 
\begin{itemize}
\item ${\tt node}(A,(G,0,\alpha),t)$ is a node of some level $1\leq i\leq\lfloor(1+\alpha)r\rfloor$ corresponding to an endpoint of an edge of $\mathcal{L}_{G}(i,i+1)\cap\overline{{\tt E}(A,(G,0,\alpha),t)}$ and
\item there is no red edge in ${\tt E}(A,(G,0,\alpha),t)$ and
\item for every layer $\mathcal{L}$ of $G$, at least half the green edges of $\mathcal{L}$ are not in ${\tt E}(A,(G,0,\alpha),t)$;
%\item if $i<\lfloor(1+\alpha)r\rfloor$, then there is a green edge $h$ of $\mathcal{L}_{G}(i,i+1)\cap\overline{{\tt E}(A,(G,0,\alpha),t)}$ and an edge of $\mathcal{L}_{G}(i+1,i+2)\cap\overline{{\tt E}(A,(G,0,\alpha),t)}$ having a common endpoint such that $h={\tt edge}(A,(G,0,\alpha),t+1)$ or no endpoints of $h$ is adjacent to an endpoint of ${\tt edge}(A,(G,0,\alpha),t)$;
\end{itemize} 
then we let $G^{new}=$ {\tt
  GraphModification}($G,\alpha,A,t$) and at least one of the
following two properties holds:
\begin{itemize}
\item $i<\lfloor(1+\alpha)r\rfloor$ and ${\tt node}(A,(G^{new},0,\alpha),t+1)$ is a node of level $i+1$ corresponding to an endpoint of an edge of $\mathcal{L}_{G^{new}}(i+1,i+2)\cap\overline{{\tt E}(A,(G^{new},0,\alpha),t+1)}$;
\item ${\tt edge}(A,(G^{new},0,\alpha),t+1)\in{\tt E}(A,(G^{new},0,\alpha),t)$.
\end{itemize}
\end{lemma}

\begin{proof}
Throughout this proof, each time we refer to a specific line, it is one of Algorithm~\ref{alg:alg1}, and thus we
omit to mention it, in order to facilitate the reading.

As in the proof of Lemma~\ref{lem:adv1}, we let $G'(j)$ be the value of variable $G'$, in the call {\tt GraphModification}($G,\alpha,A,t$), right before executing line $j$. Hence, since $G^{new}=G'(\ref{ligne:fin})$, we just have to show that at least one of the following two properties holds:

\begin{itemize}
\item $1\leq i<\lfloor(1+\alpha)r\rfloor$ and ${\tt node}(A,(G'(\ref{ligne:fin}),0,\alpha),t+1)$ is a node of level $i+1$ corresponding to an endpoint of an edge of $\mathcal{L}_{G'(\ref{ligne:fin})}(i+1,i+2)\cap\overline{{\tt E}(A,(G'({\ref{ligne:fin}}),0,\alpha),t+1)}$;
\item ${\tt edge}(A,(G'(\ref{ligne:fin}),0,\alpha),t+1)\in{\tt E}(A,(G'(\ref{ligne:fin}),0,\alpha),t)$.
\end{itemize}

Observe that if ${\tt edge}(A,(G,0,\alpha),t+1)\in {\tt E}(A,(G,0,\alpha),t)$, then $G'(\ref{ligne:fin})=G$ and we have 
${\tt edge}(A,(G'(\ref{ligne:fin}),0,\alpha),t+1)\in{\tt E}(A,(G'(\ref{ligne:fin}),0,\alpha),t)$, which means that the lemma holds. Hence, we can assume that ${\tt edge}(A,(G,0,\alpha),t+1)\not\in {\tt E}(A,(G,0,\alpha),t)$. It follows, in view of the conditions of the lemma,  that the condition at line~\ref{ligne:test} evaluates to true. We have the following two claims.

\begin{claim}\label{claim:triv}
For every $1\leq j \leq \ref{ligne:fin}$, if line~j is executed, then (1) $G'(j)\in\mathcal{F}(\lfloor(1+\alpha)r\rfloor+1,16k,r)$, (2) there is no red edge in ${\tt E}(A,(G'(j),0,\alpha),t)$, and (3) for every $1\leq h\leq\lfloor(1+\alpha)r\rfloor$, the number of green edges in $\mathcal{L}_{G'(j)}(h,h+1)\cap\overline{{\tt E}(A,(G'(j),0,\alpha),t)}$ is at least $2$.
\end{claim}

\begin{proofclaim}
Assume that line~j is executed. In Algorithm~\ref{alg:alg1}, the modification of variable $G'$ are made only with the operations ${\tt switch\mbox{-}ports}$, ${\tt move\mbox{-}gadgets}$ and ${\tt switch\mbox{-}edges}$. Hence, by Lemma~\ref{lem:modifGraph}, we have $G'(j)\in\mathcal{F}(\lfloor(1+\alpha)r\rfloor+1,16k,r)$, which proves the first part. 

Concerning the second part, assume by contradiction that there is a red edge in  ${\tt E}(A,(G'(j),0,\alpha),t)$. It follows that there is an integer $t'\leq t$ such that  ${\tt node}(A,(G'(j),0,\alpha),t')$ is a gadget. Since the set of gadgets' labels in $G$ is the same as in $G'(j)$ and ${\tt M}(A,(G,0,\alpha),t)={\tt M}(A,(G'(j),0,\alpha),t)$ by Lemma~\ref{lem:adv1}, we know that ${\tt node}(A,(G,0,\alpha),t')$ is a gadget and ${\tt edge}(A,(G,0,\alpha),t')$ is a red edge. This contradicts the fact there is no red edge in ${\tt E}(A,(G,0,\alpha),t)$ and proves the second part.

Now, suppose by contradiction that the third part does not hold. This means there is an integer $1\leq h'\leq\lfloor(1+\alpha)r\rfloor$ such that the number of green edges in $\mathcal{L}_{G'(j)}(h',h'+1)\cap\overline{{\tt E}(A,(G'(j),0,\alpha),t)}$ is at most $1$. By Lemma~\ref{lem:conse}, we know that the number of green edges in
  $\mathcal{L}_{G'(j)}(h',h'+1)$ is equal to the number of green edges in $\mathcal{L}_{G}(h',h'+1)$, namely $8k^2$. By assumption, we know that at most half the green edges of $\mathcal{L}_{G}(h',h'+1)$ are in ${\tt E}(A,(G,0,\alpha),t)$, which means that the number of green edges in $\mathcal{L}_{G}(h',h'+1)\cap\overline{{\tt E}(A,(G,0,\alpha),t)}$ is at least $4k^2\geq 4$. Consequently, there is at least one more green edge in $\mathcal{L}_{G'(j)}(h',h'+1)\cap{\tt
    E}(A,(G'(j),0,\alpha),t)$ than in
  $\mathcal{L}_{G}(h',h'+1)\cap{\tt
    E}(A,(G,0,\alpha),t)$. However, note that
  the set of labels of the nodes at level $h'$ (resp. $h'+1$) in
  $G$ is identical to the set of labels of the nodes at level
  $h'$ (resp. $h'+1$) in $G'(j)$. Also note that each green edge of
  $\mathcal{L}_{G}(h',h'+1)$
  (resp. $\mathcal{L}_{G'(j)}(h',h'+1)$) is incident to a node of
  level $h'$ and to a node of level $h'+1$ in $G$
  (resp. $G'(j)$). Hence, the fact that there is at least one more
  green edge in $\mathcal{L}_{G'(j)}(h',h'+1)\cap{\tt
    E}(A,(G'(j),0,\alpha),t)$  than in
   $\mathcal{L}_{G}(h',h'+1)\cap{\tt
    E}(A,(G,0,\alpha),t)$ implies
  there exists an integer $t'\leq t$ such that the label of ${\tt
    node}(A,(G,0,\alpha),t')$ is different from the label of
  ${\tt node}(A,(G'(j),0,\alpha),t')$. This contradicts the fact
  that ${\tt M}(A,(G,0,\alpha),t)={\tt
    M}(A,(G'(j),0,\alpha),t)$ by Lemma~\ref{lem:adv1}, which
  proves the third part, and thus the claim.
\end{proofclaim}

\begin{claim}\label{claim:green_edge}
 % At the end of the execution of
Right after line~\ref{ligne:modif:2:fin}, $u'$ is a node of level $i+1$. Moreover, ${\tt edge}(A,(G'_{\ref{ligne:test:u}},0,\alpha),t+1)$ is a green edge of $\mathcal{L}_{G'_{\ref{ligne:test:u}}}(i,i+1)$.
\end{claim}

\begin{proofclaim}
Recall that ${\tt node}(A,(G,0,\alpha),t)$ is a node of some level $1\leq i\leq\lfloor(1+\alpha)r\rfloor$ and the labels of the nodes of level $i$ in $G'(j)$, if line~j is executed, are the same as those of level $i$ in $G$ (as $G'(j)\in\mathcal{F}(\lfloor(1+\alpha)r\rfloor+1,16k,r)$ by Claim~\ref{claim:triv}). Thus, since by Lemma~\ref{lem:adv1} we have ${\tt M}(A,(G,0,\alpha),t)={\tt M}(A,(G'(j),0,\alpha),t)$, we know that ${\tt node}(A,(G'(j),0,\alpha),t)$ is also a node of level $i$ (again if line~j is executed).

Note that $G=G'(\ref{ligne:test:f})$. Besides, by assumption, ${\tt node}(A,(G,0,\alpha),t)$ is a node of some level $1\leq i\leq\lfloor(1+\alpha)r\rfloor$ corresponding to an endpoint of an edge of $\mathcal{L}_{G}(i,i+1)\cap\overline{{\tt E}(A,(G,0,\alpha),t)}$. Hence, if the condition at line~\ref{ligne:test:f} of Algorithm~\ref{alg:alg1} is false, we immediately know that variable $e$ is an edge of $\mathcal{L}_{G'(\ref{ligne:test:k})}(i,i+1)$ when testing the condition at line~\ref{ligne:test:k}. However, this also holds even if the condition at line~\ref{ligne:test:f} of Algorithm~\ref{alg:alg1} is true, in view of the definition of ${\tt switch\mbox{-}ports}$, lines~\ref{ligne:modif:1} and~\ref{ligne:modif:e:1} and the fact that ${\tt node}(A,(G'(\ref{ligne:modif:e:1}),0,\alpha),t)$ is a node of level $i$. Consequently, if $e$ is not a red edge when testing the condition at line~\ref{ligne:test:k}, we know that $e$ is green, which means that ${\tt edge}(A,(G'({\ref{ligne:test:u}}),0,\alpha),t+1)$ is a green edge of $\mathcal{L}_{G'({\ref{ligne:test:u}})}(i,i+1)$.

Otherwise, $e$ is a red edge when testing the condition at line~\ref{ligne:test:k} and we know that this condition evaluates to true because there is at least one green edge in $\mathcal{L}_{G'({\ref{ligne:test:k}})}(i,i+1)\cap\overline{{\tt E}(A,(G'({\ref{ligne:test:k}}),0,\alpha),t)}$ by Claim~\ref{claim:triv}. Thus, in view of the definition of ${\tt move\mbox{-}gadgets}$, lines~\ref{ligne:modif:2:deb} and~\ref{ligne:modif:e:2} and the fact that ${\tt node}(A,(G'(\ref{ligne:modif:2:fin}),0,\alpha),t)$ is a node of level $i$, we again have the guarantee that ${\tt edge}(A,(G'({\ref{ligne:test:u}}),0,\alpha),t+1)$ is a green edge of $\mathcal{L}_{G'({\ref{ligne:test:u}})}(i,i+1)$.

Finally, since ${\tt node}(A,(G'({\ref{ligne:test:u}}),0,\alpha),t)$ is a node of level $i$ and ${\tt edge}(A,(G'({\ref{ligne:test:u}}),0,\alpha),t+1)$ is a green edge of $\mathcal{L}_{G'({\ref{ligne:test:u}})}(i,i+1)$, it means that ${\tt node}(A,(G'({\ref{ligne:test:u}}),0,\alpha),t+1)$ is necessarily of level $i+1$, i.e., variable $u'$ is a node of level $i+1$ at the end of the execution of line~\ref{ligne:modif:2:fin}. This ends the proof of the claim.
\end{proofclaim}
~\\

Assume that $i=\lfloor(1+\alpha)r\rfloor$. In this case, the condition at line~\ref{ligne:test:u} evaluates to false and the algorithm terminates without further modifications. Hence, in view of Claim~\ref{claim:green_edge} and line~\ref{ligne:modif:2:fin}, we know that ${\tt node}(A,(G'({\ref{ligne:fin}}),0,\alpha),t+1)$ is a node of level $i+1=\lfloor(1+\alpha)r\rfloor+1$. We also know by Claims~\ref{claim:triv} and~\ref{claim:green_edge} that there is no red edge in ${\tt E}(A,(G'(\ref{ligne:test:u}),0,\alpha),t+1)$ and thus in ${\tt E}(A,(G'(\ref{ligne:fin}),0,\alpha),t+1)$. Let $H$ be the subgraph of $G'({\ref{ligne:fin}})$ induced by the edges in ${\tt E}(A,(G'({\ref{ligne:fin}}),0,\alpha),t+1)$. Since, there is no red edge in ${\tt E}(A,(G'(\ref{ligne:fin}),0,\alpha),t+1)$, $H$ does not contain the critical node (visiting the critical node would require traversing a red edge beforehand) and thus, by Lemma~\ref{lem:ecc}, the distance in $H$ from ${\tt node}(A,(G'({\ref{ligne:fin}}),0,\alpha),t+1)$ to the source node is at least $\lfloor(1+\alpha)r\rfloor+1$. This violates the distance constraint and contradicts the definition of algorithm $A$ because $(G'(\ref{ligne:fin}),0,\alpha)\in\mathcal{I}(\alpha,r)$ as $G'(\ref{ligne:fin})\in\mathcal{F}(\lfloor(1+\alpha)r\rfloor+1,16k,r)$. As a result, we can assume that $i<\lfloor(1+\alpha)r\rfloor$.

Consider the execution of the algorithm from the start of line~\ref{ligne:test:u}. If at the beginning of line~\ref{ligne:test:u}, variable $u'$ is an endpoint of an edge of $\mathcal{L}_{G'({\ref{ligne:test:u}})}(i+1,i+2)\cap\overline{{\tt E}(A,(G'({\ref{ligne:test:u}}),0,\alpha),t)}$ then, in view of Claim~\ref{claim:green_edge} and lines~\ref{ligne:modif:2:fin} and~\ref{ligne:test:u}, the algorithm terminates with ${\tt node}(A,(G'(\ref{ligne:fin}),0,\alpha),t+1)$ corresponding to an endpoint of an edge of $\mathcal{L}_{G'(\ref{ligne:fin})}(i+1,i+2)\cap\overline{{\tt E}(A,(G'(\ref{ligne:fin}),0,\alpha),t+1)}$, which proves the lemma. Hence, we can assume that when testing the condition  of line~\ref{ligne:test:u}, variable $u'$ is not an endpoint of an edge $\in\mathcal{L}_{G'}(i+1,i+2)\cap\overline{{\tt E}(A,(G',0,\alpha),t)}$. Observe that in this case, the condition of line~\ref{ligne:test:u} necessarily evaluates to true since $u'$ is a node of level $i+1$ at the end of the execution of line~\ref{ligne:modif:2:fin} (cf. Claim~\ref{claim:green_edge}) and $i<\lfloor(1+\alpha)r\rfloor$.

Now, we want to show that the condition at line~\ref{ligne:test:h} evaluates to true. To do so, we first prove below that in $G'(\ref{ligne:test:u})$ there exists a gadget $g$ adjacent to a node of $N_i$ and to a node of $N_{i+1}$. 

Observe that the endpoint $u$ (resp. $u'$) of $e={\tt edge}(A,(G'_{\ref{ligne:test:u}},0,\alpha),t+1)$ in level $i$ (resp. $i+1$) is incident to exactly $\left\lfloor\frac{w}{4}\right\rfloor=4k$ edges in $\mathcal{L}_{G'(\ref{ligne:test:u})}(i,i+1)$ (note that in view of Claim~\ref{claim:green_edge} and line~\ref{ligne:modif:2:fin}, $u$ is necessarily ${\tt node}(A,(G'(\ref{ligne:test:u}),0,\alpha),t)$). The $4k$ edges of $\mathcal{L}_{G'(\ref{ligne:test:u})}(i,i+1)$ incident to $u$ (resp. $u'$) are made only of green edges connecting $u$ (resp. $u'$) to nodes of level $i+1$ (resp. $i$) and of red edges incident to gadgets, each of which is adjacent to exactly one node of level $i+1$ (resp. $i$). Hence, in $G'(\ref{ligne:test:u})$, the number of nodes of level $i+1$ (resp. $i$) that are adjacent to a node of $D$ is at most $4k$. Denote by $V_i$ (resp. $V_{i+1}$) the set of nodes at level $i$ (resp. $i+1$) in $G'(\ref{ligne:test:u})$.

%since the sum of the number green edges connecting $u$ (resp. $u'$) to a node of level $i+1$ (resp. $i$) with the number of gadgets adjacent to both $u$ (resp. $u'$) and one node of level $i+1$ (resp. $i$) is $4k$.

It follows that $|V_i\setminus N_i|\leq 4k$. 
Observe that $U:=\mathcal{L}_{G'(\ref{ligne:test:u})}(i+1,i+2)\cap {\tt E}(A,(G'(\ref{ligne:test:u}),0,\alpha),t)$ is a subset of the green edges of $\mathcal{L}_{G'(\ref{ligne:test:u})}(i+1,i+2)$ since red edges are not in ${\tt E}(A,(G'(\ref{ligne:test:u}),0,\alpha),t)$ by Claim~\ref{claim:triv}. In view of Lemma~\ref{lem:conse}, it follows that $|U|\leq 8k^2$. Since in $G'(\ref{ligne:test:u})$ every node of level $i+1$ is connected to exactly $4k$ edges in  $\mathcal{L}_{G'(\ref{ligne:test:u})}(i+1,i+2)$, there can only be at most $2k$ nodes of level $i+1$ in $G'(\ref{ligne:test:u})$ that are not incident to an edge of $\mathcal{L}_{G'(\ref{ligne:test:u})}(i+1,i+2)\setminus U$. Hence, from what we explained above, we necessarily have $|V_{i+1}\setminus N_{i+1}|\leq 4k +2k = 6k$.

Using Lemma~\ref{lem:conse}, we can show there are exactly $56k^2$ gadgets adjacent to a node of level $i$ and a node of level $i+1$. Observe that at most $\left\lfloor\frac{w}{4}\right\rfloor 4k =16k^2$ of these gadgets are connected to a node in $V_i\setminus N_i$ and at most $\left\lfloor\frac{w}{4}\right\rfloor 6k = 24k^2$ of these gadgets are connected to a node in $V_{i+1}\setminus N_{i+1}$. 
Let $q$ be the number of the gadgets connected to both a node of $N_{i}$ and a node of $N_{i+1}$.
We have :

\begin{eqnarray*}
  q & \geq & 56k^2 - 16k^2 - 24k^2\\
  & \geq & 16k^2\\
  & \geq & 1 \mbox{ since $k\geq 1$}
\end{eqnarray*}

It follows that, in $G'(\ref{ligne:test:u})$, there is at least one gadget $g$ adjacent to a node of $N_i$ and to a node of $N_{i+1}$. Moreover, by Claim~\ref{claim:triv}, there exists a green edge in $(\mathcal{L}_{G'(\ref{ligne:test:u})}(i,i+1)\cap\overline{{\tt E}(A,(G'(\ref{ligne:test:u}),0,\alpha),t)})\setminus \{e\}$. Hence, the condition at line~\ref{ligne:test:h} evaluates to true and line~\ref{ligne:modif:4} to~\ref{ligne:modif:3} are executed. 

With the {\tt move\mbox{-}gadget} operation of line~\ref{ligne:modif:4}, the gadget $g$, which is adjacent to node $v\in N_i$ and node $v'\in N_{i+1}$, and the red edges incident to $g$ are replaced by an edge $k=\{v,v'\}$, while a green edge $\{x,x'\}$ of $(\mathcal{L}_{G'(\ref{ligne:test:u})}(i,i+1)\cap\overline{{\tt E}(A,(G'(\ref{ligne:test:u}),0,\alpha),t)})\setminus \{e\}$ is replaced by a gadget $g'$ and two red edges $\{x,g'\}$ and $\{x',g'\}$. Observe that since $v\in N_i$ and $v'\in N_{i+1}$, there is no edge in $G'(\ref{ligne:test:u})$ linking $v$ or $v'$ to an endpoint of $e={\tt edge}(A,(G'_{\ref{ligne:test:u}},0,\alpha),t+1)=\{u,u'\}$ or a gadget adjacent to an endpoint of $e$. In view of the {\tt move\mbox{-}gadget} operation at line~\ref{ligne:modif:4} and the fact that ${\tt M}(A,(G'({\ref{ligne:test:u}}),0,\alpha),t)={\tt M}(A,(G'({\ref{ligne:modif:4}}),0,\alpha),t)={\tt M}(A,(G'({\ref{ligne:modif:3}}),0,\alpha),t)$ (cf. Lemma~\ref{lem:adv1}), this remains true just before executing the {\tt switch\mbox{-}edge} operation of line~\ref{ligne:modif:3}, i.e., there is no edge in $G'(\ref{ligne:modif:3})$ linking $v$ or $v'$ to an endpoint of $e={\tt edge}(A,(G'(\ref{ligne:modif:3})),0,\alpha),t+1)=\{u,u'\}$ or a gadget adjacent to an endpoint of $e$. 

Hence, by the definition of the {\tt switch\mbox{-}edge} operation, when executing line~\ref{ligne:modif:3}, edge $\{u,u'\}$ and $\{v,v'\}$ are replaced by edges $\{u,v'\}$ and $\{v,u'\}$. In particular, the port at node $u$ of edge $\{u,v'\}$ is the same as the port that had edge $\{u,u'\}$ at node $u$. Thus, since $u$ is ${\tt node}(A,(G'(\ref{ligne:test:u}),0,\alpha),t)$ and ${\tt M}(A,(G'({\ref{ligne:test:u}}),0,\alpha),t)={\tt M}(A,(G'({\ref{ligne:modif:3}}),0,\alpha),t)={\tt M}(A,(G'({\ref{ligne:fin}}),0,\alpha),t)$, it follows that $v'={\tt node}(A,(G'(\ref{ligne:fin}),0,\alpha),t+1)$. Recall that, since node $v'$ belongs to $N_{i+1}$, we know that it is a node of level $i+1$ in $G'(\ref{ligne:test:u})$ that is incident to an edge of $\mathcal{L}_{G'(\ref{ligne:test:u})}(i+1,i+2)\cap\overline{{\tt E}(A,(G'(\ref{ligne:test:u}),0,\alpha),t)}$. Note that the operation {\tt move\mbox{-}gadget}  (resp. {\tt switch\mbox{-}edge}), executed at line~\ref{ligne:modif:4} (resp. line~\ref{ligne:modif:3}), never affected edges belonging to the layer $\mathcal{L}_{G'(\ref{ligne:modif:4})}(i+1,i+2)$ (resp. $\mathcal{L}_{G'(\ref{ligne:modif:3})}(i+1,i+2)$), the nodes' labels, or the nodes' degrees. Moreover, the set of labels of the nodes at level $i+1$ in
  $G'(\ref{ligne:test:u})$ is identical to the set of labels of the nodes at level
  $i+1$ in $G'(\ref{ligne:fin})$ because $G'(\ref{ligne:test:u})$ and $G'(\ref{ligne:fin})$ both belong to $\mathcal{F}(\lfloor(1+\alpha)r\rfloor+1,16k,r)$ by Claim~\ref{claim:triv}. Therefore, $v'={\tt node}(A,(G'(\ref{ligne:fin}),0,\alpha),t+1)$ is a node of level $i+1$ corresponding to an endpoint of an edge of $\mathcal{L}_{G'(\ref{ligne:fin})}(i+1,i+2)\cap\overline{{\tt E}(A,(G'(\ref{ligne:fin}),0,\alpha),t+1)}$. This ends the proof of this lemma.

\end{proof}

Roughly speaking, the next lemma gives a lower bound on the penalty incurred by a distance-constrained exploration algorithm when running in some graphs of $\mathcal{F}$, by the time of the first visit of a gadget. Those graphs are determined by function {\tt AdversaryBehavior}, the execution of which corresponds to the first stage of our overall strategy outlined in Section~\ref{sec:int}.

\begin{lemma}
\label{lem:penagadget}
Let $\alpha$ be any positive real, $r\geq 6$ be any integer and $A(\alpha,r)$ be any algorithm solving distance-constrained exploration for all instances of $\mathcal{I}(\alpha,r)$. Let $G={\tt AdversaryBehavior}(r,\alpha,A,16k)$, where $k$ is any positive integer. We have the following two properties:  
\begin{enumerate}
\item $(G,0,\alpha)$ is an instance of $\mathcal{I}(\alpha,r)$ such that $G\in\mathcal{F}(\lfloor(1+\alpha)r\rfloor+1,16k,r)$;
\item The penalty incurred by $A$ on instance $(G,0,\alpha)$, by the time of the first visit of a gadget, is at least $k^2$.
\end{enumerate}
\end{lemma}

\begin{proof}
The execution of function {\tt AdversaryBehavior}$(r,\alpha,A,16k)$ can be viewed as a sequence of steps 0,1,2, etc. Precisely, in step $0$, (cf. line~\ref{ligne:alg0:1} of Algorithm~\ref{alg:alg0}), we choose any graph $G_0=(V_0,E_0)\in\mathcal{F}(\lfloor(1+\alpha)r\rfloor+1,16k,r)$. By definition, $(G_0,0,\alpha)$ is then necessarily an instance of $\mathcal{I}(\alpha,r)$. In step $x\geq1$ (cf. lines~\ref{ligne:alg0:2} to~\ref{ligne:alg0:last} of Algorithm~\ref{alg:alg0}), we check whether $x< {\tt cost}(A,(G_x,0,\alpha))$ or not, where $G_x=(V_x,E_x)$ is the graph returned by ${\tt GraphModification}(G_{x-1},\alpha,A,x-1)$. If $x<{\tt cost}(A,(G_x,0,\alpha))$, we proceed to step $x+1$. Otherwise, we definitively stop (step $x$ is then the last step) and function {\tt AdversaryBehavior}$(r,\alpha,A,16k)$ returns graph $G_x$. Note that, instance $(G_x,0,\alpha)$ is solvable by $A(\alpha,r)$ (and thus function ${\tt cost}$ is used with well-defined parameters when checking whether $x<{\tt cost}(A,(G_x,0,\alpha))$) because using the first property of Lemma~\ref{lem:adv1} and the fact that $(G_0,0,\alpha)\in \mathcal{I}(\alpha,r)$, it can be shown by induction on $x$ that $(G_x,0,\alpha)\in\mathcal{I}(\alpha,r)$. Using the same property, we can also prove by induction on $x$ that $G_x\in\mathcal{F}(\lfloor(1+\alpha)r\rfloor+1,16k,r)$, which particularly means that $G_x$ has the same size and the same order as $G_0$, i.e., $|V_0|=|V_x|$ and $|E_0|=|E_x|$. Hence, since $A(\alpha,r)$ is an algorithm solving distance-constrained exploration for all instances of $\mathcal{I}(\alpha,r)$, the behavior of the adversary is necessarily made of a finite number of steps. Therefore, there exists a last step, the index of which will be denoted by $\tau$. In light of the above explanations, we know that function {\tt AdversaryBehavior}$(r,\alpha,A,16k)$ returns graph $G_{\tau}$ and $G_{\tau}$ (resp. $(G_{\tau},0,\alpha)$) belongs to $\mathcal{F}(\lfloor(1+\alpha)r\rfloor+1,16k,r)$ (resp. $\mathcal{I}(\alpha,r)$), which means that the first property of the lemma holds.

The purpose of the rest of this proof is to show that second property also holds, i.e., the penalty incurred by $A(\alpha,r)$ on instance $(G_{\tau},0,\alpha)$, by the time of the first visit of a gadget, is at least $k^2$.

%Consider an adversary applying in steps the following behavior. In step $0$, it simply chooses any instance $(G_0,0,\alpha)$ of $\mathcal{I}(\alpha,r)$ such that $G_0=(V_0,E_0)\in\mathcal{F}(\lfloor(1+\alpha)r\rfloor+1,16k,r)$. In step $i\geq1$, the adversary emulates algorithm $A(\alpha,r)$ on instance $(G_i,0,\alpha)$, where $G_i=(V_i,E_i)$ is the graph returned by ${\tt GraphModification}(G_{i-1},\alpha,A,i)$, and checks whether ${\tt cost}(A,(G_i,0,\alpha))> i$ or not: if it is the case, it proceeds to step $i+1$, otherwise it definitively stops (and thus step $i$ is the last step). Note that, instance $(G_i,0,\alpha)$ is solvable by $A(\alpha,r)$ (and thus the emulation is consistent) because using the first property of Lemma~\ref{lem:adv1} it can be shown by induction on $i$ that $(G_i,0,\alpha)\in\mathcal{I}(\alpha,r)$. Using the same property, we can also prove by induction on $i$ that $G_i\in\mathcal{F}(\lfloor(1+\alpha)r\rfloor+1,16k,r)$, which particularly means that $G_i$ has the same size and the same order as $G_0$, i.e., $|V_0|=|V_i|$ and $|E_0|=|E_i|$. Hence, since $A(\alpha,r)$ is an algorithm solving distance-constrained exploration for all instances of $\mathcal{I}(\alpha,r)$, the behavior of the adversary is necessarily made of a finite number of steps. Therefore, there exists a last step, the index of which will be denoted by $\tau$.
%The purpose of the rest of this proof is to show that the penalty incurred by $A(\alpha,r)$ on instance $(G_{\tau},0,\alpha)$, by the time of the first visit of a gadget, is at least $k^2$.

Consider an execution $\mathcal{E}$ of algorithm $A(\alpha,r)$ on instance $(G_{\tau},0,\alpha)$. Let $\lambda$ be the first time in execution $\mathcal{E}$ when half the green edges of some layer have been completely traversed at least once. Such a layer corresponds to $\mathcal{L}_{G_{\tau}}(j,j+1)$ for some integer $1\leq j\leq \lfloor(1+\alpha)r\rfloor$ because, by construction, every graph of $\mathcal{F}(\lfloor(1+\alpha)r\rfloor+1,16k,r)$ is made of exactly $\lfloor(1+\alpha)r\rfloor$ layers. Denote by ${\tt DOWN}$ (resp. ${\tt UP}$) the number of times the agent has traversed, during the time interval $[0,\lambda]$ in execution $\mathcal{E}$, a green edge of $\mathcal{L}_{G_{\tau}}(j,j+1)$ from a node of level $j$ (resp. $j+1$) to a node of level $j+1$ (resp. $j$).

We have the following two claims.

\begin{claim}
\label{theo1:claim0}
${\tt cost}(A,(G_{\tau},0,\alpha))=\tau$. Moreover, for every $0\leq x<\tau$, we have ${\tt cost}(A,(G_x,0,\alpha))>x$ and ${\tt M}(A,(G_x,0,\alpha),x)={\tt M}(A,(G_\tau,0,\alpha),x)$.
\end{claim}

\begin{proofclaim}
In view of lines~\ref{ligne:alg0:22} to~\ref{ligne:alg0:last} of Algorithm~\ref{alg:alg0}, we immediately know that ${\tt cost}(A,(G_{\tau},0,\alpha))\leq \tau$ and, for every $0\leq x <\tau$, ${\tt cost}(A,(G_{x},0,\alpha))>x$. This particularly means that ${\tt cost}(A,(G_{\tau-1},0,\alpha))>\tau-1$: after the first $\tau-1$ edges traversals in $G_{\tau-1}$, the algorithm asks the agent to make a $\tau$th edge traversal by taking some port $p$ from some node of label $l$. Since $A$ is deterministic and, by Lemma~\ref{lem:adv1},  ${\tt M}(A,(G_{\tau-1},0,\alpha),\tau-1)={\tt M}(A,(G_{\tau},0,\alpha),\tau-1)$, it follows that ${\tt cost}(A,(G_{\tau},0,\alpha))> \tau-1$, i.e., with instance $(G_{\tau},0,\alpha)$, algorithm $A$ also asks the agent to take port $p$ from a node of label $l$ after the first $\tau-1$ edges traversals in $G_{\tau}$. This implies that ${\tt cost}(A,(G_{\tau},0,\alpha))= \tau$. Thus, it just remains to show that for every $0\leq x<\tau$, ${\tt M}(A,(G_{x},0,\alpha),x)={\tt M}(A,(G_{\tau},0,\alpha),x)$. We can do this using a descending induction.

We mentioned above that ${\tt M}(A,(G_{\tau-1},0,\alpha),\tau-1)={\tt M}(A,(G_{\tau},0,\alpha),\tau-1)$, meaning that the base case $x=\tau-1$ immediately holds. Moreover, if for an integer $0<x\leq\tau-1$, we have ${\tt M}(A,(G_{x},0,\alpha),x)={\tt M}(A,(G_{\tau},0,\alpha),x)$, then we have  ${\tt M}(A,(G_{x-1},0,\alpha),x-1)={\tt M}(A,(G_{\tau},0,\alpha),x-1)$. Indeed, ${\tt M}(A,(G_{x},0,\alpha),x)={\tt M}(A,(G_{\tau},0,\alpha),x)$ implies that ${\tt M}(A,(G_{x},0,\alpha),x-1)={\tt M}(A,(G_{\tau},0,\alpha),x-1)$ and, by Lemma~\ref{lem:adv1}, we have ${\tt M}(A,(G_{x-1},0,\alpha),x-1)={\tt M}(A,(G_{x},0,\alpha),x-1)$. This proves the inductive step, and concludes the proof of this claim.
\end{proofclaim}

\begin{claim}
\label{theo1:claim0bis}
For every integer $0\leq x<\tau$ and for every integer $1\leq h \leq  \lfloor(1+\alpha)r\rfloor$, the number of green edges in $\mathcal{L}_{G_{x}}(h,h+1)\cap{\tt E}(A,(G_{x},0,\alpha),x)$ (resp. $\mathcal{L}_{G_{x}}(h,h+1)\cap\overline{{\tt E}(A,(G_{x},0,\alpha),x)}$) is equal to the number of green edges in $\mathcal{L}_{G_{\tau}}(h,h+1)\cap{\tt E}(A,(G_{\tau},0,\alpha),x)$ (resp. $\mathcal{L}_{G_{\tau}}(h,h+1)\cap\overline{{\tt E}(A,(G_{\tau},0,\alpha),x)}$)
\end{claim}

\begin{proofclaim}
  Suppose, for the sake of contradiction, that the claim does not
  hold. First, recall that
  $G_x\in\mathcal{F}(\lfloor(1+\alpha)r\rfloor+1,16k,r)$, for every
  integer $0\leq x \leq\tau$. So, by Lemma~\ref{lem:conse}, we know that
  for every integer $0\leq x<\tau$ and for every integer $1\leq h \leq
  \lfloor(1+\alpha)r\rfloor$, the number of green edges in
  $\mathcal{L}_{G_{x}}(h,h+1)$ is equal to the number of green edges
  in $\mathcal{L}_{G_{\tau}}(h,h+1)$. Since we assume that the claim
  does not hold, this necessarily implies that there exist an integer
  $0\leq x'<\tau$ and an integer $1\leq h' \leq
  \lfloor(1+\alpha)r\rfloor$ such that the number of green edges in
  $\mathcal{L}_{G_{x'}}(h',h'+1)\cap{\tt E}(A,(G_{x'},0,\alpha),x')$
  is not equal to the number of green edges in
  $\mathcal{L}_{G_{\tau}}(h',h'+1)\cap{\tt
    E}(A,(G_{\tau},0,\alpha),x')$. This means there is at least one
  more green edge in $\mathcal{L}_{G_{x'}}(h',h'+1)\cap{\tt
    E}(A,(G_{x'},0,\alpha),x')$ than in
  $\mathcal{L}_{G_{\tau}}(h',h'+1)\cap{\tt
    E}(A,(G_{\tau},0,\alpha),x')$ or vice versa. However, note that
  the set of labels of the nodes at level $h'$ (resp. $h'+1$) in
  $G_{\tau}$ is identical to the set of labels of the nodes at level
  $h'$ (resp. $h'+1$) in $G_{x'}$. Also note that each green edge of
  $\mathcal{L}_{G_{\tau}}(h',h'+1)$
  (resp. $\mathcal{L}_{G_{x'}}(h',h'+1)$) is incident to a node of
  level $h'$ and to a node of level $h'+1$ in $G_{\tau}$
  (resp. $G_{x'}$). Hence, the fact that there is at least one more
  green edge in $\mathcal{L}_{G_{x'}}(h',h'+1)\cap{\tt
    E}(A,(G_{x'},0,\alpha),x')$ than in
  $\mathcal{L}_{G_{\tau}}(h',h'+1)\cap{\tt
    E}(A,(G_{\tau},0,\alpha),x')$ or vice versa, necessarily implies
  there exists an integer $x''\leq x'$ such that the label of ${\tt
    node}(A,(G_{x'},0,\alpha),x'')$ is different from the label of
  ${\tt node}(A,(G_{\tau},0,\alpha),x'')$. This contradicts the fact
  that ${\tt M}(A,(G_{x'},0,\alpha),x')={\tt
    M}(A,(G_\tau,0,\alpha),x')$ (cf. Claim~\ref{theo1:claim0}), which
  proves the claim.
\end{proofclaim}

~~\\
In the rest of this proof, we will denote by $\Gamma$ the set of all the gadgets' labels in $G_{\tau}$. Since for any two graphs of $\mathcal{F}(\lfloor(1+\alpha)r\rfloor+1,16k,r)$, their gadgets' labels are identical, we know that, for every $0\leq x <\tau$, the set of all the gadgets' labels in $G_{x}$ is also $\Gamma$. We continue with the following two claims.

\begin{claim}
\label{theo1:claim1}
During the time interval $[0,\lambda]$ in execution $\mathcal{E}$, the agent never occupies a gadget of $G_{\tau}$.
\end{claim}

\begin{proofclaim}
According to the definition of time $\lambda$, we know that at this time, the agent completes the $s$th edge traversal of execution $\mathcal{E}$, for some integer $1\leq s \leq\tau$ and this traversal corresponds to a traversal of a green edge of $G_{\tau}$. Therefore, at time $\lambda$, the agent cannot occupy a gadget. Moreover, at time $0$, the agent occupies the source node, i.e., the node having label $0$ in $G_{\tau}$, which is not a gadget. Finally, since the neighbors of the source node of $G_{\tau}$ all belong to level $1$, ${\tt node}(A,(G_{\tau},0,\alpha),1)$ is necessarily a node of this level and, thus, cannot be a gadget either.

As a result, to prove the claim, we only need to show that for every integer $2\leq z<s$, ${\tt node}(A,(G_{\tau},0,\alpha),z)$ is a not a gadget. Assume for the sake of contradiction that there exists an integer $2\leq z<s$ such that ${\tt node}(A,(G_{\tau},0,\alpha),z)$ is a gadget. Also assume, without loss of generality, that $z$ is the smallest integer for which this occurs. Since the source node and ${\tt node}(A,(G_{\tau},0,\alpha),1)$ are not gadgets, it follows that before the end of the $z$th edge traversal in execution $\mathcal{E}$, no gadget is visited. In particular, this trivially implies that ${\tt node}(A,(G_{\tau},0,\alpha),z-1)$ cannot be a gadget. But it also implies that ${\tt node}(A,(G_{\tau},0,\alpha),z-1)$ cannot be the critical node of $G_{\tau}$ or a node belonging to the tail of $G_{\tau}$. Indeed, if it was not true, then the agent would have necessarily occupied a gadget before making the $z$th edge traversal of execution $\mathcal{E}$ because all paths from the source node to the critical node or a node of the tail pass through a gadget, which would be a contradiction. 

Consequently, ${\tt node}(A,(G_{\tau},0,\alpha),z-1)$ is either the source node of $G_{\tau}$ or a node belonging to some level in $G_{\tau}$. In the first case, we know that ${\tt node}(A,(G_{\tau},0,\alpha),z)$ is a node of level $1$ because the neighbors of the source node of $G_{\tau}$ all belong to this level. However, this contradicts the assumption that ${\tt node}(A,(G_{\tau},0,\alpha),z)$ is a gadget. Hence, ${\tt node}(A,(G_{\tau},0,\alpha),z-1)$ is necessarily a node belonging to some level $p$ in $G_{\tau}$. Moreover, the labels of the nodes of level $p$ in $G_{\tau}$ are the same as those of level $p$ in $G_{z-1}$ and, by Claim~\ref{theo1:claim0}, ${\tt M}(A,(G_{z-1},0,\alpha),z-1)={\tt M}(A,(G_\tau,0,\alpha),z-1)$. Thus ${\tt node}(A,(G_{z-1},0,\alpha),z-1)$ is a node of level $p$ in $G_{z-1}$.

Recall that, in execution $\mathcal{E}$, the end of the $s$th edge traversal, which occurs at time $\lambda$, is the first time when half the green edges of some layer in $G_{\tau}$ have been completely traversed at least once. Thus, since $2\leq z<s$, we know by Lemma~\ref{lem:conse} that for every integer $1\leq h \leq  \lfloor(1+\alpha)r\rfloor$, the number of green edges in $\mathcal{L}_{G_{\tau}}(h,h+1)\cap\overline{{\tt E}(A,(G_{\tau},0,\alpha),z-1)}$ is at least $4k^2>1$. By Claim~\ref{theo1:claim0bis}, this means that for every integer $1\leq h \leq  \lfloor(1+\alpha)r\rfloor$ there is at least one green edge in $\mathcal{L}_{G_{z-1}}(h,h+1)\cap\overline{{\tt E}(A,(G_{z-1},0,\alpha),z-1)}$. 

Also recall that before the end of the $z$th edge traversal in execution $\mathcal{E}$, no gadget is visited i.e., no node with a label in $\Gamma$ is visited. Since ${\tt M}(A,(G_{z-1},0,\alpha),z-1)={\tt M}(A,(G_\tau,0,\alpha),z-1)$ (by Claim~\ref{theo1:claim0}), and since all gadgets' labels in $G_{z-1}$ also belong to $\Gamma$ and each red edge of $G_{z-1}$ is incident to a gadget, it follows that there is no red edge in ${\tt E}(A,(G_{z-1},0,\alpha),z-1)$.

As a result, we can apply Lemma~\ref{lem:adv1} to state that ${\tt node}(A,(G_z,0,\alpha),z)$ is not a gadget of $G_z$. This means that the label of ${\tt node}(A,(G_z,0,\alpha),z)$ does not belong to $\Gamma$. However, by Claim~\ref{theo1:claim0}, ${\tt M}(A,(G_z,0,\alpha),z)={\tt M}(A,(G_\tau,0,\alpha),z)$. Hence, the label of ${\tt node}(A,(G_{\tau},0,\alpha),z)$ cannot belong to $\Gamma$ either, which implies that ${\tt node}(A,(G_{\tau},0,\alpha),z)$ is not a gadget of $G_{\tau}$. This is a contradiction with the assumption made above that ${\tt node}(A,(G_{\tau},0,\alpha),z)$ is a gadget, which concludes the proof of this claim.
\end{proofclaim}

By Claim~\ref{theo1:claim1}, we know that during the time interval $[0,\lambda]$ in execution $\mathcal{E}$, the agent never occupies a gadget of $G_{\tau}$, and thus, by extension, the critical node. In view of the construction of family $\mathcal{F}$, this implies that after each traversal of a green edge of $\mathcal{L}_{G_{\tau}}(j,j+1)$ from a node of level $j$ (resp. $j+1$) to a node of level $j+1$ (resp. $j$) completed at some time $\beta<\lambda$ during execution $\mathcal{E}$, the agent can come back to a node level $j$ (resp. $j+1$) in the time interval $(\beta,\lambda]$ only by completing the traversal of a green edge of $\mathcal{L}_{G_{\tau}}(j,j+1)$ from a node of level $j+1$ (resp. $j$) in the time interval $(\beta,\lambda]$. From this, it follows that ${\tt DOWN}={\tt UP}+ \epsilon$ for some $\epsilon\in[-1..1]$. Moreover, by Lemma~\ref{lem:conse}, we know that the number of green edges in every layer of $G_{\tau}$ is $8k^2$. Consequently, at time $\lambda$, exactly $4k^2$ green edges of $\mathcal{L}_{G_{\tau}}(j,j+1)$ have been entirely explored. This means that ${\tt DOWN}+{\tt UP}=4k^2$ and thus ${\tt DOWN}\geq 2k^2-1$. Since $k 
\geq 1$, we can state that  ${\tt DOWN}-1\geq k^2$. Hence, to finish the proof of the second property of the lemma, it is enough to show the claim below.

\begin{claim}
\label{theo1:claim3}
The penalty incurred by algorithm $A(\alpha,r)$ on instance $(G_{\tau},0,\alpha)$ by time $\lambda$ is at least  ${\tt DOWN}-1$.
\end{claim}

\begin{proofclaim}
%Need of Claims~\ref{theo1:claim0} and~\ref{theo1:claim1}
%Recall that, in execution $\mathcal{E}$, $\lambda$ is the first time when half the green edges of some layer in $G_{\tau}$ have been completely traversed at least once, and this layer is denoted by $\mathcal{L}_{G_{\tau}}(j,j+1)$. 
In the proof of this claim, we will refer to each traversal in execution $\mathcal{E}$ of a previously unexplored green edge of $\mathcal{L}_{G_{\tau}}(j,j+1)$ from a node of level $j$ to a node of level $j+1$ as a \emph{descending traversal}. We will say that two descending traversals are {\em consecutive} if there is no other descending traversal made between them (however, other types of traversals could still occur between two descending traversals that are consecutive).

Note that there are ${\tt DOWN}$ descending traversals that are completed by time $\lambda$ in execution $\mathcal{E}$, and that ${\tt DOWN}>k^2\geq 1$ (as ${\tt DOWN}-1\geq k^2$ and $k\geq1$). Therefore, the claim holds if we have the following property: for every pair of positive integers $\varphi<\varphi'$ such that the $\varphi$th and $\varphi'$th edge traversals in execution $\mathcal{E}$ are two consecutive descending traversals completed by time $\lambda$, there exists an integer $\varphi\leq\varphi''<\varphi'$ such that the $\varphi''$th edge traversal in execution $\mathcal{E}$ passes through a previously explored edge. So, assume, for the sake of contradiction, there exist two positive integers $\varphi<\varphi'$ such that the $\varphi$th and $\varphi'$th edge traversals in execution $\mathcal{E}$ are two consecutive descending traversals completed by time $\lambda$, and there is no integer $\varphi\leq\varphi''<\varphi'$ for which the $\varphi''$th edge traversal in execution $\mathcal{E}$ passes through a previously explored edge. Since ${\tt cost}(A,(G_{\tau},0,\alpha))=\tau$ by Claim~\ref{theo1:claim0}, we have $\varphi'\leq\tau$. Moreover, $\varphi\geq2$ since, by definition, the edges of $G_{\tau}$ connected to the node labeled $0$ are not green and do not belong to any layer.

Note that, by Claim~\ref{theo1:claim1}, during the time interval $[0,\lambda]$, the agent never occupies a gadget of $G_{\tau}$ in execution $\mathcal{E}$, i.e., it never occupies a node with a label belonging to $\Gamma$. Hence, for every integer $\mu<\varphi'$, there is no red edge in ${\tt E}(A,(G_{\mu},0,\alpha),\mu)$ because ${\tt M}(A,(G_{\mu},0,\alpha),\mu)={\tt M}(A,(G_\tau,0,\alpha),\mu)$ (by Claim~\ref{theo1:claim0}), all gadgets' labels in $G_{\mu}$ also belong to $\Gamma$, and each red edge of $G_{\mu}$ is incident to a gadget. Also note that, in view of Claim~\ref{theo1:claim0bis} and the fact that in each layer of $G_{\tau}$ at most half the green edges have been explored by time $\lambda$ in execution $\mathcal{E}$, we can state that at least half the green edges of each layer of $G_{\mu}$ are not in ${\tt E}(A,(G_{\mu},0,\alpha),\mu)$, for every integer $\mu<\varphi'$. Finally, observe that ${\tt node}(A,(G_{\tau},0,\alpha),\varphi-1)$ and ${\tt node}(A,(G_{\tau},0,\alpha),\varphi'-1)$ (resp. ${\tt node}(A,(G_{\tau},0,\alpha),\varphi)$ and ${\tt node}(A,(G_{\tau},0,\alpha),\varphi')$) are nodes of level $j$ (resp. $j+1$) in view of the definitions of $\varphi$ and $\varphi'$. Thus, ${\tt node}(A,(G_{\varphi-1},0,\alpha),\varphi-1)$ and ${\tt node}(A,(G_{\varphi'-1},0,\alpha),\varphi'-1)$ (resp. ${\tt node}(A,(G_{\varphi},0,\alpha),\varphi)$ and ${\tt node}(A,(G_{\varphi'},0,\alpha),\varphi')$) are nodes of level $j$ (resp. $j+1$) because
the labels of the nodes of level $j$ (resp. $j+1$) in $G_{\tau}$ are the same as those of level $j$ (resp. $j+1$) in $G_{x}$, for every $0\leq x<\tau$, and because, by Claim~\ref{theo1:claim0},  ${\tt M}(A,(G_{x},0,\alpha),x)={\tt M}(A,(G_\tau,0,\alpha),x)$, still for every $0\leq x<\tau$. 

We have two cases to consider. The first case is when, for every integer $\varphi\leq\mu<\varphi'$, ${\tt node}(A,(G_{\mu},0,\alpha),\mu)$ is a node of some level $1\leq h\leq\lfloor(1+\alpha)r\rfloor$ corresponding to an endpoint of an edge of $\mathcal{L}_{G_\mu}(h,h+1)\cap\overline{{\tt E}(A,(G_\mu,0,\alpha),\mu)}$. In this case, based on the observations given in the previous paragraph and the fact that $2\leq\varphi\leq\mu<\varphi'\leq\tau$, we can use Lemma~\ref{lem:adv2} to prove by induction on $\mu$ that one of the following two properties is true: (1) for every $\varphi\leq\mu<\varphi'$, ${\tt node}(A,(G_{\mu+1},0,\alpha),\mu+1)$ is a node of level $j+1+\mu-\varphi$ or (2) for some $\varphi\leq\mu<\varphi'$, ${\tt node}(A,(G_{\mu},0,\alpha),\mu)$ is a node of level $j+\mu-\varphi$ and ${\tt edge}(A,(G_{\mu+1},0,\alpha),\mu+1)\in{\tt E}(A,(G_{\mu+1},0,\alpha),\mu)$. If the first property is true then ${\tt node}(A,(G_{\varphi'-1},0,\alpha),\varphi'-1)$ is a node of some level $\sigma>j$, which contradicts what we have stated above i.e., ${\tt node}(A,(G_{\varphi'-1},0,\alpha),\varphi'-1)$ is a node of level $j$. If the second property is true then, for some $\varphi\leq\mu<\varphi'$, ${\tt node}(A,(G_{\tau},0,\alpha),\mu)$ is a node of level $j+\mu-\varphi$ and ${\tt edge}(A,(G_{\tau},0,\alpha),\mu+1)\in{\tt E}(A,(G_{\tau},0,\alpha),\mu)$ because the labels of the nodes of level $j+\mu-\varphi$ in $G_{\mu}$ are the same as those of level $j+\mu-\varphi$ in $G_{\tau}$ and because, by Claim~\ref{theo1:claim0}, ${\tt M}(A,(G_{\mu},0,\alpha),\mu)={\tt M}(A,(G_\tau,0,\alpha),\mu)$ and ${\tt M}(A,(G_{\mu+1},0,\alpha),\mu+1)={\tt M}(A,(G_\tau,0,\alpha),\mu+1)$. As explained below, this leads to a contradiction, whether $\mu+1=\varphi'$ or not.
\begin{itemize}
\item If $\mu+1=\varphi'$ then 
${\tt node}(A,(G_{\tau},0,\alpha),\mu)={\tt node}(A,(G_{\tau},0,\alpha),\varphi'-1)$ and thus ${\tt node}(A,(G_{\tau},0,\alpha),\mu)$ is a node of level $j$. Since we stated above that ${\tt node}(A,(G_{\tau},0,\alpha),\mu)$ is a node of level $j+\mu-\varphi$, it follows that $\mu=\varphi$. However, we also stated earlier that ${\tt node}(A,(G_{\tau},0,\alpha),\varphi)$ is a node of level $j+1$. Hence, ${\tt node}(A,(G_{\tau},0,\alpha),\mu)$ is both at levels $j$ and $j+1$, which is a contradiction.

\item If $\mu+1\ne\varphi'$, we precisely have $\varphi<\mu+1<\varphi'$ and, in execution $\mathcal{E}$, the $(\mu+1)$th edge traversal passes through a previously explored edge, which is again a contradiction.
\end{itemize}

As a result, the first case cannot occur.

The second case, which is complementary to the first one, is when for some $\varphi\leq\mu<\varphi'$ ${\tt node}(A,(G_{\mu},0,\alpha),\mu)$ is not a node of some level $1\leq h\leq\lfloor(1+\alpha)r\rfloor$ corresponding to an endpoint of an edge of $\mathcal{L}_{G_\mu}(h,h+1)\cap\overline{{\tt E}(A,(G_\mu,0,\alpha),\mu)}$. Without loss of generality, assume that $\mu$ is the smallest integer at least equal to $\varphi$ for which such a situation occurs. To analyze this second case, we first suppose that $\varphi<\mu$. This means that $3\leq\mu<\tau$ as $\varphi\geq2$ and $\mu<\varphi'\leq\tau$. In view of the assumption of minimality of $\mu$, the fact that $\varphi<\mu$ implies that ${\tt node}(A,(G_{\mu-1},0,\alpha),\mu-1)$ is a node of some level $1\leq h\leq\lfloor(1+\alpha)r\rfloor$ corresponding to an endpoint of an edge of $\mathcal{L}_{G_{\mu-1}}(h,h+1)\cap\overline{{\tt E}(A,(G_{\mu-1},0,\alpha),\mu-1)}$. Moreover, according to the explanations given just before the analysis of the first case, we also know that there is no red edge in ${\tt E}(A,(G_{\mu-1},0,\alpha),\mu-1)$ and at least half the green edges of each layer of $G_{\mu-1}$ are not in ${\tt E}(A,(G_{\mu-1},0,\alpha),\mu-1)$. Hence, by Lemma~\ref{lem:adv2}, either (1) ${\tt node}(A,(G_{\mu},0,\alpha),\mu)$ is a node of level $1\leq h+1\leq\lfloor(1+\alpha)r\rfloor$ corresponding to an endpoint of an edge of $\mathcal{L}_{G_\mu}(h+1,h+2)\cap\overline{{\tt E}(A,(G_\mu,0,\alpha),\mu)}$ or (2) ${\tt edge}(A,(G_{\mu},0,\alpha),\mu)\in{\tt E}(A,(G_{\mu},0,\alpha),\mu-1)$. The first situation leads to an immediate contradiction with the definition of the currently analyzed case. In the second situation, we can state that ${\tt edge}(A,(G_{\tau},0,\alpha),\mu)\in{\tt E}(A,(G_{\tau},0,\alpha),\mu-1)$ because, by Claim~\ref{theo1:claim0}, ${\tt M}(A,(G_{\mu},0,\alpha),\mu)={\tt M}(A,(G_\tau,0,\alpha),\mu)$. However, this means that, in execution $\mathcal{E}$, the $\mu$th edge traversal passes through a previously explored edge, while $\varphi\leq\mu<\varphi'$: this is also a contradiction. Thus, to complete the analysis of the second case and thereby prove this claim, suppose now that $\varphi=\mu$. This means that $2\leq\mu<\tau$.

According to what we have stated earlier, we know that there is no red edge in ${\tt E}(A,(G_{\mu-1},0,\alpha),\mu-1)$ and at least half the green edges of each layer of $G_{\mu-1}$ are not in ${\tt E}(A,(G_{\mu-1},0,\alpha),\mu-1)$. We also know that ${\tt node}(A,(G_{\varphi-1},0,\alpha),\varphi-1)$, which corresponds to ${\tt node}(A,(G_{\mu-1},0,\alpha),\mu-1)$, is a node of level $j$ in $G_{\varphi-1}=G_{\mu-1}$. Hence, if ${\tt node}(A,(G_{\mu-1},0,\alpha),\mu-1)$ is an endpoint of an edge of $\mathcal{L}_{G_{\mu-1}}(j,j+1)\cap\overline{{\tt E}(A,(G_{\mu-1},0,\alpha),\mu-1)}$, we can apply Lemma~\ref{lem:adv2} to establish that: (1) ${\tt node}(A,(G_{\mu},0,\alpha),\mu)$ is a node of level $1\leq j+1\leq\lfloor(1+\alpha)r\rfloor$ corresponding to an endpoint of an edge of $\mathcal{L}_{G_\mu}(j+1,j+2)\cap\overline{{\tt E}(A,(G_\mu,0,\alpha),\mu)}$ or (2) ${\tt edge}(A,(G_{\mu},0,\alpha),\mu)\in{\tt E}(A,(G_{\mu},0,\alpha),\mu-1)$. Using the same arguments as those used when $\varphi<\mu$, these two situations lead to contradictions. Consequently, ${\tt node}(A,(G_{\mu-1},0,\alpha),\mu-1)$ is not an endpoint of an edge of $\mathcal{L}_{G_{\mu-1}}(j,j+1)\cap\overline{{\tt E}(A,(G_{\mu-1},0,\alpha),\mu-1)}$. This means that the $4k$ edges of $\mathcal{L}_{G_{\mu-1}}(j,j+1)$ that are incident to ${\tt node}(A,(G_{\mu-1},0,\alpha),\mu-1)$  all belong to ${\tt E}(A,(G_{\mu-1},0,\alpha),\mu-1)$. Keep in mind that the set of labels of the nodes at level $j$ (resp. $j+1$) in $G_{\tau}$ is identical to the set of labels of the nodes at level $j$ (resp. $j+1$) in $G_{\mu-1}$. Also recall that each green edge of $\mathcal{L}_{G_{\tau}}(j,j+1)$ (resp. $\mathcal{L}_{G_{\mu-1}}(j,j+1)$) is incident to a node of level $j$ and to a node of level $j+1$ in $G_{\tau}$ (resp. $G_{\mu-1}$). As a result,  since ${\tt M}(A,(G_{\mu-1},0,\alpha),\mu-1)={\tt M}(A,(G_\tau,0,\alpha),\mu-1)$ (cf. Claim~\ref{theo1:claim0}), we know that the $4k$ edges of $\mathcal{L}_{G_{\tau}}(j,j+1)$ that are incident to ${\tt node}(A,(G_{\tau},0,\alpha),\mu-1)={\tt node}(A,(G_{\tau},0,\alpha),\varphi-1)$ all belong to ${\tt E}(A,(G_{\tau},0,\alpha),\mu-1)={\tt E}(A,(G_{\tau},0,\alpha),\varphi-1)$. This implies that the $\varphi$th edge traversal in execution $\mathcal{E}$ (from ${\tt node}(A,(G_{\tau},0,\alpha),\varphi-1)$ to ${\tt node}(A,(G_{\tau},0,\alpha),\varphi)$) corresponds to a traversal of a previously explored edge, which is a contradiction. This concludes the analysis of the second case and, by extension, the proof of this claim.
\end{proofclaim}

\end{proof}

We are now ready to prove the following theorem. The proof is based on a construction that corresponds to the second stage of our overall strategy outlined in Section~\ref{sec:int}.

\begin{theorem}
\label{theo:theo1}
Let $\alpha$ be any positive real and let $r\geq 6$ be any integer. Let $A(\alpha,r)$ be an algorithm solving distance-constrained exploration for all instances of $\mathcal{I}(\alpha,r)$. For every integer $k>1$, there exists an instance $(G=(V,E),l_s,\alpha)$ of $\mathcal{I}(\alpha,r)$ such that $|V|\geq k$ and for which the penalty of $A$ is at least $\left(\frac{|V|}{16\left(2\left\lfloor(1+\alpha)r\right\rfloor+3\right)}\right)^2$.
\end{theorem}

\begin{proof}
From Lemma~\ref{lem:penagadget}, for every integer $k>1$, there exists an instance $(G=(V,E),0,\alpha)$ of $\mathcal{I}(\alpha,r)$ such that $G\in\mathcal{F}(\lfloor(1+\alpha)r\rfloor+1,16k,r)$ and for which the penalty already incurred by $A$, by the time $\lambda$ of the first visit of a gadget, is at least $k^2$. For ease of reading, we let $w=16k$ and $\ell=\lfloor(1+\alpha)r\rfloor+1$.

For any $1\leq i\leq \ell-1$, we consider the $\left\lfloor \frac{w}{4}\right\rfloor$-regular bipartite graph $B_i$ between vertices of $V_i$ and $V_{i+1}$ that is used in the first step of the construction of $G$ (cf. Section~\ref{sec:preli}). Such a graph can be obtained by taking the subgraph of $G$ induced by $V_i \cup V_{i+1}$ and replacing, for any $1\leq j\leq \left\lfloor\frac{7\beta}{8}\right\rfloor$ (with $\beta=w\left\lfloor\frac{w}{4}\right\rfloor$), the gadget $g^i_j$ and the two edges connecting $g^i_j$ to its neighbor $x_j^i$ in $V_i$ and its neighbor $y_j^{i+1}$ in $V_{i+1}$, by an edge $e_j^i$ connecting $x^i_j$ and $y^{i+1}_j$. The graph $B_i$ admits a proper edge 
coloring with $\left\lfloor \frac{w}{4}\right\rfloor=4k$ colors by the König's Theorem \cite{konig1916} since it is $4k$-regular bipartite graph. For any $1\leq j\leq \left\lfloor\frac{7\beta}{8}\right\rfloor$ and $1\leq i \leq \ell-1$, let $c_j^i$ be the color associated to $e_j^i$ in the proper edge coloring of $B_i$.

For a set of vertex $S$ of a graph $H$, we define the operation of \emph{merging} vertices of $S$ in one vertex $s$ as the operation consisting of removing all vertices of $S$ and all incident edges to a vertex of $S$, then adding a vertex $s$ with an unused label (so that the graph remains consistently labeled), adding an edge $e'$ corresponding to edge $e$ connecting $u$ to $s$ for any edge $e$ that connected $u$ to a vertex of $S$ in $H$ and assigning to $e'$ the same port number as $e$ for $u$ and any valid port number for $s$. For $1\leq p\leq 4k$, let $P_p=\{g_j^i\ |\ c_j^i=p \mbox{ and } i \equiv 0 \mod{2}\}$ and $I_p=\{g_j^i\ |\ c_j^i=p \mbox{ and } i \equiv 1 \mod{2}\}$. We construct graph $G'=(V',E')$ from $G=(V,E)\in\mathcal{F}(\lfloor(1+\alpha)r\rfloor+1,16k,r)$ by merging, for each $1\leq p\leq 4k$, vertices in $P_p$ (resp. $I_p$) in one vertex $x_p$ (resp. $y_p$) and then removing all edges but one between $x_p$ (resp. $y_p$) and the critical node $c(v_2)$ where $c$ is the map $c:V\rightarrow V'$ associating any vertex $u\in V$ to its corresponding vertex, i.e., either the same vertex in $V'$ if $u$ is not a gadget or the vertex $x_p$ (resp. $y_p$) if $u$ was a gadget in $P_p$ (resp. $I_p$) for some $p$. We extend this map $c$ to edges by setting $c(\{u,v\}):=\{c(u),c(v)\}$. For each node $u\in V$, the label of $c(u)\in V'$ is the same as node $u$ if $u$ is not a gadget and any valid label otherwise. We first need to show that graph $G'$ is simple.

\begin{claim}
The graph $G'$ is simple.
\end{claim}

\begin{proofclaim}
It suffices to show that for any pair of vertices $g_j^i, g_{j'}^{i'} \in P_p$ (resp. $I_p$), $g_j^i$ and $g_{j'}^{i'}$ do not have a common neighbor except the critical node. We consider two cases, depending on whether $i=i'$ or not. In the case where $i=i'$, edges $e_j^i$ and $e_{j'}^{i'}=e_{j'}^{i}$ are in the same graph $B_i$. Since $c_j^i=c_{j'}^{i'}$, edges $e_{j}^{i}$ and $e_{j'}^{i'}$ do not have a common endpoint. Hence, vertex $g_{j}^{i}$ and $g_{j'}^{i'}$ do not have a common neighbor except the critical node since their neighbors are endpoints of these two edges. In the case where $i\neq i'$, since $i$ and $i'$ have the same parity, we can assume w.l.o.g. that $i+2\leq i'$. Now, $g_j^i$ is adjacent to vertices in $V_i\cup V_{i+1}$ and $g_{j'}^{i'}$ is adjacent to vertices in $V_{i'}\cup V_{i'+1}$ that is disjoint to $V_i\cup V_{i+1}$ since $i+2\leq i'$.
\end{proofclaim}

We need to prove the following claim in order to show that $(G'=(V',E'),0,\alpha)$ is an instance of $\mathcal{I}(\alpha,r)$.

\begin{claim}
  The eccentricity of node of label $0$ in $G'$ is $r$.
\end{claim}

\begin{proofclaim}
Observe that for any path $P=(\{u_1,u_2\},\{u_2,u_3\}\dots, \{u_{t-1},u_t\})$ in $G$ there exists a path $P'=(\{c(u_1),c(u_2)\},\{c(u_2),c(u_3)\}\dots, \{c(u_{t-1}),c(u_t)\})$ in $G'$ since for any edge $\{u_i,u_{i+1}\}\in E$, there exists an edge $\{c(u_i),c(u_{i+1})\}\in E'$. Hence, for all $u,v\in V(G)$, we know that the distance from $u$ to $v$ in $G$ is at least the distance from $c(u)$ to $c(v)$ in $G'$. The eccentricity of node $c(v_1)$ of label $0$ in $G'$ is at most $r$ since the eccentricity of $v_1$ is $r$ in $G$ by Lemma~\ref{lem:ecc}. It remains to show that the eccentricity of node $c(v_1)$ in $G'$ is at least $r$. Recall that in $G$ all paths from the critical node $v_2$ to $v_1$ have length at least 3 since they pass through a gadget and a node of $V_1$. It follows that, in $G'$, all paths from $c(v_2)$ to $c(v_1)$ have length equal to 3 since they pass through a merged gadget $x_p$ for some $p$ and a node of $V_1$. Since the nodes of $L$ are unchanged in $G'$, the tail tip of $L$ is exactly at distance $r - 3$ from $c(v_2)$ and at distance $r$ from $c(v_1)$.
\end{proofclaim}

Let $t$ be the number of edge traversals performed by $A$ before the time $\lambda$ of the visit of the first gadget.
\begin{claim}\label{claim:same_behavior_merged_graph}
For every $0\leq i<t$, we have ${\tt M}(A,(G,0,\alpha),i)={\tt M}(A,(G',0,\alpha),i)$, 
${\tt node}(A,(G',0,\alpha),i)=c({\tt node}(A,(G,0,\alpha),i))$  and ${\tt E}(A,(G',0,\alpha),i)=c({\tt E}(A,(G,0,\alpha),i))$.
\end{claim}

\begin{proofclaim}
We show by induction on $i$ the statement of the claim. First, we consider the execution of the algorithm $A$ in both graphs $G$ and $G'$ before the first traversal of an edge. We have ${\tt node}(A,(G',0,\alpha),0)=c({\tt node}(A,(G,0,\alpha),0))$ since $v_1$ has label $0$ in $G$ and $c(v_1)$ has label $0$ in $G'$. We have ${\tt M}(A,(G,0,\alpha),0)={\tt M}(A,(G',0,\alpha),0)$ since the local view of the agent at node $v_1$ and $c(v_1)$ is the same in both graphs $G$ and $G'$ as both vertices have the same port numbering around them. We have ${\tt E}(A,(G',0,\alpha),i)=\emptyset=c({\tt E}(A,(G,0,\alpha),i))$ and so the statement of the claim is true for $i=0$. 

Assume the induction holds for $i-1$. Consider the execution of the algorithm $A$ in both graphs $G$ and $G'$ before the traversal of the $i$th edge. By the induction hypothesis, the memory of the agent is the same in both configurations and the agent is at node $u$ in $G$ and node $c(u)$ in $G'$. Hence, the agent will choose the same port number in both configurations and traverse $e = {\tt edge}(A,(G,0,\alpha),i)$ in $G$ and $e'={\tt edge}(A,(G',0,\alpha),i)$. Since $i<t$, no endpoint of $e$ is a gadget and so we have $e'=c(e)$. It follows that ${\tt E}(A,(G',0,\alpha),i)= c(e) \cup {\tt E}(A,(G',0,\alpha),i-1) = c(e) \cup c({\tt E}(A,(G,0,\alpha),i-1)) = c({\tt E}(A,(G,0,\alpha),i))$. After, the traversal of edge $e:=\{u,v\}$, the agent arrives at node $v$ in $G$ and node $c(v)$ in $G'$. Hence, we have ${\tt node}(A,(G',0,\alpha),i)=c({\tt node}(A,(G,0,\alpha),i))$. Since $v$ is not a gadget, the local view of the agent is the same at node $v$ in $G$ and $c(v)$ in $G'$. It follows that the memory of the agent will be the same in both configurations. 
\end{proofclaim}

By Claim~\ref{claim:same_behavior_merged_graph}, the penalty incurred by algorithm $A(\alpha, r)$ on instance $(G', 0, \alpha)$ by time $\lambda$ is the same as the penalty incurred by algorithm $A(\alpha, r)$ on instance $(G, 0, \alpha)$ by time $\lambda$. Observe that in order to construct $G'$ from $G$, we have merged the $(\ell -1)\left\lfloor\frac{7\beta}{8}\right\rfloor$ gadgets of $G$ into $2(4k)$ vertices. Hence, we have:

\begin{eqnarray*}
|V'| & = & |V| - (\ell -1)\left\lfloor\frac{7\beta}{8}\right\rfloor + 2(4k)\\
|V'| & = & w\ell+ (\ell-1)\left\lfloor\frac{7\beta}{8}\right\rfloor+r-1 - (\ell -1)\left\lfloor\frac{7\beta}{8}\right\rfloor + 8k \mbox{~~(by Lemma~\ref{lem:conse})}\\
|V'| & = & 16k\left(\left\lfloor(1+\alpha)r\right\rfloor+1\right)+r-1  + 8k\\
|V'| & = & 8k\left(2\left\lfloor(1+\alpha)r\right\rfloor+3\right)+r-1
\end{eqnarray*}

We have:

\begin{eqnarray*}
k & = & \frac{|V'|-r+1}{8\left(2\left\lfloor(1+\alpha)r\right\rfloor+3\right)} \\
k^2 & = & \left(\frac{|V'|-r+1}{8\left(2\left\lfloor(1+\alpha)r\right\rfloor+3\right)}\right)^2
\end{eqnarray*}

Hence, the penalty incurred by $A(\alpha, r)$ on instance $(G', 0, \alpha)$ is at least
$\left(\frac{|V'|}{16\left(2\left\lfloor(1+\alpha)r\right\rfloor+3\right)}\right)^2$ since $|V'|\geq 2(r-1)$.

\end{proof}

Note that in the statement of Theorem~\ref{theo:theo1}, the instance $(G=(V, E), l_s, \alpha)$ can always be chosen so that $r^2\in o(|V|)$. Therefore, we have the following corollary.

\begin{corollary}
\label{col:col1}
Let $\alpha$ be any positive real and let $r\geq 6$ be any integer. There exists no algorithm solving distance-constrained exploration for all instances $(G=(V,E),l_s,\alpha)$ of $\mathcal{I}(\alpha,r)$ with a linear penalty in $|V|$.
\end{corollary}

\section{Impossibility Result for Fuel-constrained Exploration}

In this section, we observe that a linear penalty in $|V|$ cannot be obtained for the task of fuel-constrained exploration either. Actually, using much simpler arguments, we can quickly get a lower bound stronger than the one established for distance-constrained exploration in Theorem~\ref{theo:theo1}. As for the previous variant, the tuple $(G=(V,E),l_s,\alpha)$ can consistently define an instance of the fuel-constrained exploration problem. Hence, we can reuse it analogously, along with the set $\mathcal{I}(\alpha,r)$, in the statements of the following theorem and corollary dedicated to fuel-constrained exploration.

\begin{theorem}
\label{theo:theo2}
Let $\alpha$ be any positive real and let $r\geq 2$ be any integer such that $r\alpha\geq 1$. Let $A(\alpha,r)$ be an algorithm solving fuel-constrained exploration for all instances of $\mathcal{I}(\alpha,r)$. For every positive integer $k$, there exists an instance $(G=(V,E),l_s,\alpha)$ of $\mathcal{I}(\alpha,r)$ such that $|V|\geq k$ and for which the penalty of $A$ is at least $\frac{|V|^2}{8\alpha}$.
\end{theorem}

\begin{proof}
  Let $G=(V,E)$ be a $\Big([(8\alpha+6)k-1]r+1,r-1\Big)$-lollipop
  graph, i.e., $G$ has $(8\alpha+6)rk$ nodes and consists of a line
  $L$ of $r-1$ nodes and a clique $C$ of $|C| = |V| - (r-1)$ nodes
  connected with a bridge. We designate the endpoint of $L$ that is
  not connected to $C$ as the source node $s$. Finally, we assign
  labels and port numbers in any way ensuring that the resulting graph
  $G=(V,E)$ is consistently labeled. Notice that $|V| \geq k \geq 1$
  and, by construction, the eccentricity of $s$ is $r$.

  Consider now an execution of the fuel-constrained exploration algorithm $A(\alpha,r)$ with the instance $(G=(V,E),l_s,\alpha)$, where $l_s$ is the label of node $s$. The tail of $G$ is defined as the subgraph of $G$ made of the $r-1$ edges that are not in $C$, and the set of nodes that are incident to these edges in $G$, including the source node.
  As $C$ contains $\frac{|C|\cdot (|C|-1)}{2}$ edges,  the agent has to entirely traverse the tail of $G$ at least $2\cdot \frac{|C|\cdot (|C|-1)}{4\alpha r}-1$ times during the execution. Hence, the penalty $P$ of $A$ is at least $\left(\frac{|C|\cdot (|C|-1)}{2\alpha r}-2\right) \cdot (r-1)$. Now, since $r \geq 2$, we have

  \begin{align*}
   % P & \geq (\frac{|C|\cdot (|C|-1)}{2\alpha r}-2) \cdot \frac{r}{2}\\
    P & \geq \frac{|C|\cdot (|C|-1)-4\alpha r}{2\alpha r} \cdot \frac{r}{2}\\
    P & \geq \frac{|C|\cdot (|C|-1)-4\alpha r}{4\alpha}\\
    P & \geq \frac{(|V|-r+1)\cdot (|V|-r)-4\alpha r}{4\alpha}\\
    P & \geq \frac{|V|^2-2r|V|+r^2+|V|-(4\alpha+1)r}{4\alpha}\\
    \intertext{Since $r$ and $\alpha$ are positive, and $|V|\geq 1$, we have}
    P & \geq \frac{|V|^2-2r|V|-(4\alpha+1)r|V|}{4\alpha}\\
    P & \geq \frac{|V|^2-(4\alpha+3)r|V|}{4\alpha}\\
  \intertext{Since $\frac{V^2}{2}-(4\alpha+3)r|V| \geq 0$ for $|V| \geq (8\alpha+6)r$ and $|V| = (8\alpha+6)rk$ with $k \geq 1$, we have}
    P & \geq \frac{|V|^2}{8\alpha}\\
\end{align*}
This concludes the proof of this theorem.
\end{proof}

From Theorem~\ref{theo:theo2}, we immediately get the following corollary.

\begin{corollary}
\label{col:col2}
Let $\alpha$ be any positive real and let $r\geq 2$ be any integer such that $r\alpha\geq 1$. There exists no algorithm solving fuel-constrained exploration for all instances $(G=(V,E),l_s,\alpha)$ of $\mathcal{I}(\alpha,r)$ with a linear penalty in $|V|$.
\end{corollary}

\section{Conclusion}
In this paper, we addressed open problems posed by Duncan et al. in~\cite{DuncanKK06}, who asked whether distance-constrained or fuel-constrained exploration could always be achieved with a penalty of \(O(|V|)\). For each of these two exploration variants, we completely solved the open problem by providing a strong negative answer, which in turn highlights a clear difference with the task of unconstrained exploration that can always be conducted with a penalty of $\mathcal{O}(|V|)$ (cf. \cite{PanaiteP99}).

It should be noted that we did not make any attempt at optimizing the lower bounds on the penalties: our only concern was
% just (un peu trop appuyé ..)
to prove their nonlinearity with respect to $|V|$. However, for the case of fuel-constrained exploration, there is not much room for improvement: we showed that a penalty of \(o(|V|^2)\) cannot be guaranteed, which is almost optimal as it is known that the problem can always be solved with a penalty of \(O(|V|^2)\) (cf.~\cite{DuncanKK06}). 

In contrast, such an observation cannot be made for the case of distance-constrained exploration. Indeed, the best-known algorithm for solving distance-constrained exploration can sometimes entail a penalty of $\Theta(|V|^2)$ (cf.~\cite{DuncanKK06}), while we showed that the penalty for this variant can sometimes be of order $\frac{|V|^2}{r^2}$. Hence, this leaves fully open the intriguing question of the optimal bound on the penalty for the task of distance-constrained exploration.

\bibliographystyle{plain}
\bibliography{biblio}

\end{document}